\documentclass[opre,nonblindrev]{informs4C}

\OneAndAHalfSpacedXI

\usepackage{lmodern}
\usepackage{mathtools}
\usepackage{amsmath}
\usepackage{amssymb}
\usepackage{mathrsfs}
\usepackage{bm}
\usepackage{booktabs}
\usepackage{enumitem}
\usepackage{algorithm}
\usepackage{algorithmic}
\usepackage{float}
\usepackage{microtype}
\usepackage[final,hypertexnames=false]{hyperref}
\usepackage{natbib}

\bibpunct[, ]{(}{)}{,}{a}{}{,}%
\def\bibfont{}%

\hypersetup{colorlinks=true,linkcolor=blue,citecolor=blue,urlcolor=blue}

\newtheorem{definition}{Definition}
\newtheorem{lemma}{Lemma}
\newtheorem{proposition}{Proposition}
\newtheorem{theorem}{Theorem}
\newtheorem{corollary}{Corollary}
\newtheorem{remark}{Remark}

\ECRepeatTheorems
\EquationsNumberedThrough
\MANUSCRIPTNO{}

\newcommand{\E}{\mathbb E}
\newcommand{\Prb}{\mathbb P}

\newcommand{\one}{\mathbf 1}
\newcommand{\R}{\mathbb R}
\newcommand{\cM}{\mathcal M}
\newcommand{\cT}{\mathcal T}
\newcommand{\cN}{\mathcal N}
\newcommand{\cA}{\mathcal A}

\newcommand{\cQ}{\mathcal Q}
\newcommand{\SBLP}{\operatorname{SBLP}_{\rm MC}}

\definecolor{olive}{RGB}{128,128,0}
\definecolor{firebrick}{RGB}{178,34,34}

\begin{document}

\RUNTITLE{Multi-Capacity Online Assortment Optimization under Markov Choice}
\TITLE{Multi-Capacity Online Assortment Optimization with Consecutive Stays under Markov Chain Choice}

\ARTICLEAUTHORS{
\AUTHOR{Hongbin Zhang}
\AFF{Sauder School of Business, University of British Columbia, BC V6T 1Z4, Canada} 
\AUTHOR{Yanzhi Li}
\AFF{College of Business, City University of Hong Kong, Tat Chee Avenue, Kowloon, Hong Kong}
\AUTHOR{Woonghee Tim Huh}
\AFF{Sauder School of Business, University of British Columbia, BC V6T 1Z4, Canada}
}

\ABSTRACT{%
We study online assortment optimization under Markov chain choice
and independent Bernoulli arrivals, with each period's request
type known in advance. Each product may have multiple
interchangeable physical copies, and each booking reserves
consecutive service slots on a single copy. Customers choose
among products, while the platform assigns purchases to feasible
copies. We develop a sales-based relaxation that
upper-bounds the optimal online revenue and use it to design
a policy with two components: independent sampling and
purchase--update coupling.
We generate independent copy proposals and make a product
eligible whenever at least one of its copies proposes.
A general revenue-extraction bound for random candidate sets
with independent product inclusions shows that optimizing over eligible products recovers at least half of
each period's relaxation revenue benchmark in expectation.
A purchase--update coupling reproduces the interval flows
through independent virtual copy updates while preserving
physical feasibility. Under general Markov chain choice,
the policy uses a random gate to make the coupling feasible:
it offers the selected assortment with probability \(1-1/e\)
and the empty assortment otherwise.
No gate is needed under multinomial logit (MNL) choice.

The resulting policies achieve approximation ratios of
\((1-1/e)/2\) and \(1/2\), respectively, relative to the optimal
nonanticipative online policy, independently of product capacities.
When outside-option attractiveness is zero in every period,
the same ungated MNL policy achieves a \(1-1/e\) guarantee.
Both policies admit exact implementations in expected
polynomial time. In the single-copy MNL setting, our policy
improves the guarantee of \citet{huwu2026} from \(0.271\)
to \(1/2\) under the same arrival convention.
The \(1/2\) static extraction factor and the \(1-1/e\)
uniform scalar-gate factor are separately sharp, and the
\(1/2\) MNL guarantee is tight for the proposed policy.
}


\maketitle

\section{Introduction}
\label{sec:intro}

Advance booking changes the nature of capacity. In our setting,
a sale does not simply remove one interchangeable unit from an
inventory pool; it reserves a particular physical unit for the
entire requested service interval, with the assignment fixed
irrevocably at booking. In high-speed rail, for example, an
accepted itinerary may require immediate and irrevocable
assignment to a single seat. Having an available seat on each
leg is not sufficient: the same seat must be available on all
requested legs \citep{zhu2023assign}. Similar requirements arise in hotels, vehicle-rental companies, and other reservation systems when a booking must be fulfilled by the same room or asset throughout the requested stay. These requirements call
for tracking each physical unit's availability across service
slots \citep{rusmevichientong2023revenue}.

The physical unit used to fulfill a booking, however, is
often not the alternative chosen by the customer. A rail
passenger may choose a seating class, while the firm assigns
a specific seat within that class. A related form of
fulfillment flexibility appears in flexible-product models,
which allow the seller to assign a customer to one of several
specified alternatives after purchase
\citep{gallego2004managing}. We call each customer-facing
category a \emph{product} and its individual physical units
\emph{copies}. Copies of the same product are interchangeable
from the customer's perspective, but the firm must distinguish
them for assignment and capacity accounting.

This product--copy distinction matters for both feasibility
and customer choice. With existing assignments fixed, tracking
only the number of available copies in each service slot may
suggest that a request is feasible even when no single copy
can fulfill it. For example, one seat may be available on the
first leg of a journey and another on the second, even though
neither is available for both legs. Displaying physical copies
as separate alternatives without a corresponding reformulation
of the choice model can change product purchase probabilities.
We therefore model customer choice at the product level while
tracking individual copies for feasibility and assignment.

Consecutive stays also complicate the evolution of each copy's
availability. Reserving an interior portion of a free interval
splits it into two disjoint blocks, so the feasibility of later
requests depends on which service slots remain free, not merely
how many. This differs from standard inventory models, which
track remaining stock, and from many reusable-resource models,
which summarize availability through release times or counts
of available units. In our setting, booking periods and service
slots play distinct roles: requests arrive sequentially during
the booking horizon, while each accepted reservation occupies
an interval of service slots on its assigned copy. Once reserved,
those slots remain unavailable on that copy for the rest of
the booking horizon.

We adopt the fixed-type Bernoulli-arrival convention of
\citet{huwu2026} over a finite booking horizon. Each period
has at most one potential customer, and arrivals are independent
across periods. The arrival probabilities, requested intervals,
revenues, and choice parameters for all periods are known before
the booking horizon begins and may vary across periods.
Arrival and choice outcomes are observed only as they occur.
Before observing the arrival outcome in each period, the platform
selects an assortment of products with feasible copies and offers
it if the customer arrives. Each purchase is assigned to a single
copy available throughout the requested interval. We model
customer substitution using the Markov chain choice model of
\citet{blanchet2016markov}. The chain is defined over
customer-visible products and the outside option, and the
customer's path stops upon first reaching an offered product
or the outside option. This model includes multinomial logit
(MNL) choice with positive outside-option attractiveness and
allows richer substitution patterns. We treat MNL with zero
outside-option attractiveness separately.

The closest related work is \citet{huwu2026}, whose
fixed-type Bernoulli model is included in our setting when
each product has one copy and choice follows MNL.
Building on their interval-flow relaxation and
proposal-discarding framework, we extend the analysis to
general Markov chain choice and multiple copies per product.
The main challenge is to coordinate product-level customer
choice with copy-level proposals and virtual updates.

Our benchmark is the expected revenue of the optimal nonanticipative
online policy. Like our algorithm, this policy knows the model
primitives but observes arrival and choice outcomes only as they
occur. The benchmark ranges over all feasible history-dependent
policies, without computational restrictions. Approximation
guarantees for efficient policies relative to this benchmark
are often called \emph{philosopher inequalities}, whereas
prophet inequalities use a clairvoyant offline benchmark
\citep{papadimitriou2024online,huang2024online}.
We construct policies that admit exact implementations in
expected polynomial time and achieve approximation guarantees
of \((1-1/e)/2\) under general Markov chain choice and
\(1/2\) under MNL, relative to this benchmark and independently
of product capacities.

\subsection{Literature Review}
\label{subsec:literature}

Our work relates to four research streams:
(i) choice-based revenue management,
(ii) flexible assignment and reusable resources,
(iii) interval scheduling and advance booking, and
(iv) online approximation and rounding.

(i) \textbf{Choice-Based Revenue Management.}
Classical network revenue management studies how to allocate
limited capacity among products with fixed resource requirements,
balancing current revenue against the value of preserving capacity
for future demand
\citep{gallegovanryzin1997,tallurivanryzin1998}.
Choice-based models incorporate customer substitution by allowing
demand to depend on the assortment offered.
\citet{talluri2004general} study revenue management under a general
discrete-choice model. \citet{liu2008choice} analyze a choice-based
linear programming formulation for network revenue management,
and \citet{gallego2015general} develop a sales-based formulation
that incorporates customer choice into capacity allocation.

Related work examines how assortments should vary across
customers and over time.
\citet{golrezaei2014realtime} study real-time personalized
assortments, while \citet{aouad2018greedy} develop greedy-like
algorithms for assortment planning under MNL choice.
\citet{masimchilevi2020algorithms} analyze online algorithms
for assortment and pricing decisions, and
\citet{masimchilevizhao2021dynamic} study the performance
of price and assortment calendars fixed in advance.

Within this literature, the Markov chain choice model represents
substitution through transitions among product states
\citep{blanchet2016markov}.
\citet{feldman2017revenue} develop LP formulations for assortment
optimization under this model and an equivalent compact relaxation
for network revenue management. They also provide a procedure
for recovering an optimal assortment mixture.
\citet{desir2020constrained} develop algorithms for constrained
assortment optimization under the same choice model.
Our analysis uses the continuation structure of Markov paths
to establish a lower bound on the expected revenue from
optimizing over a random candidate set in which products
are included independently.

(ii) \textbf{Flexible Assignment and Reusable Resources.}
Flexible-product models allow the seller to assign a customer
to one of several specified alternatives after purchase
\citep{gallego2004managing}.
Our model involves a related form of fulfillment flexibility:
the customer purchases a product, and the seller assigns the
booking to one of its interchangeable physical copies.

A related literature studies assortment, pricing, and allocation
for resources that become available again after service completion.
\citet{rusmevichientong2020dynamic} study dynamic assortment
optimization with random usage durations.
\citet{lei2020realtime} analyze dynamic pricing with deterministic
service durations, with extensions to heterogeneous durations
and advance reservations.
\citet{gong2022online} study online assortment optimization with
reusable resources, while \citet{baekma2022bifurcating} develop
approximation methods for network revenue management with resource
reuse. Further work considers Bayesian online assortment
optimization \citep{feng2024near} and online allocation under
adversarial requests \citep{goyal2025reusable}.

Our setting combines Markov substitution among products with
the requirement that each booking be fulfilled by a single
physical copy throughout the requested service interval.
Accepted reservations remain fixed throughout the booking horizon,
and the feasibility of subsequent requests depends on which
service slots remain free on each copy. The policy must therefore
coordinate assortment and assignment decisions using each
copy's free-slot set.

(iii) \textbf{Interval Scheduling and Advance Booking.}
Classical fixed-interval scheduling assigns jobs with prescribed
start and end times to machines so that no two jobs overlap
on the same machine
\citep{kolen2007interval,kovalyov2007fixed}.
\citet{chawla2019pricing} study Bayesian online allocation with
interval demand. In revenue-management applications,
\citet{zhu2023assign} impose the assign-to-seat requirement that
an accepted rail itinerary use the same physical seat throughout
the journey. \citet{rusmevichientong2023revenue} study assortment
decisions for heterogeneous unit-capacity resources requested
over service intervals. These studies highlight the need to
track availability on individual resources when each request
must be fulfilled by the same resource throughout its service
interval.

The closest methodological reference is \citet{huwu2026}, who
develop interval-flow relaxations and use virtual resource states
with proposal discarding for consecutive-stay revenue management.
Their choice-based model uses MNL and treats each physical
resource as a customer-facing alternative. Under fixed-type
Bernoulli arrivals, this model is included in our MNL
specialization with one physical copy per product. They also
consider general arrivals with random request types.
Their MNL policy samples virtual resource updates after
observing the customer's choice. This construction ensures
that the purchased resource is updated and recovers the
prescribed independent update law after averaging over
the choice outcome.

We build on this framework to accommodate general Markov
chain choice and multiple copies per product. Our static
analysis establishes a revenue-extraction bound for random
candidate sets in which products are included independently.
For dynamic implementation, we characterize coupling feasibility
through necessary and sufficient Hall-type coverage conditions.
A priority-rule representation and copy-level lifting then
generate the required virtual updates. Under fixed-type
Bernoulli arrivals, our MNL policy improves the approximation
guarantee of \citet{huwu2026} from \(0.271\) to \(1/2\).
This improvement holds even with one copy per product and,
in that setting, follows from the stronger static
revenue-extraction bound.
Table~\ref{tab:closest-comparison} summarizes the comparison.

\begin{table}[t]
\TABLE
{Comparison under fixed-type Bernoulli arrivals.
\label{tab:closest-comparison}}
{\begin{tabular}{@{}p{0.21\textwidth}p{0.30\textwidth}p{0.35\textwidth}@{}}
\toprule
Dimension & \citet{huwu2026} & Our paper \\
\midrule
Displayed alternative
& Individual resource
& Product with interchangeable copies \\

Choice model
& MNL
& Markov chain choice \\

Proposals
& Resource-level
& Copy-level, aggregated by product \\

Revenue analysis
& MNL-specific
& Markov continuation \\

Update coupling
& Conditional resource updates
& Hall coupling with copy-level lifting \\

MNL guarantee
& \(0.271\)
& \(1/2\) \\
\bottomrule
\end{tabular}}
{Notes. Both guarantees are relative to the optimal online benchmark.
The zero-outside-option MNL case is treated separately in our paper.}
\end{table}

(iv) \textbf{Online Approximation and Rounding.}
Our approximation guarantees use the optimal online benchmark
described above. The guarantees of \citet{husekarwu2026philosopher} concern
inventory depletion and do not directly apply to our setting,
where feasibility depends on each physical copy's availability
throughout the requested service interval.
In online stochastic matching, relaxation marginals guide
randomized proposals, while discarding and dependent rounding
help preserve feasibility and control dependence
\citep{braverman2022max,braverman2025new}.
In our setting, implementing the relaxation requires tracking
each copy's free service slots and assigning each purchase
to a copy available throughout the requested interval. 

Contention-resolution schemes similarly transform fractional
or independently sampled solutions into feasible discrete outcomes
\citep{chekuri2014submodular,kashaev2023uniform,
sergeev2024contention}.
Our implementation must also preserve a prescribed joint law
for virtual copy updates. After observing the purchase outcome,
the policy generates updates that include at least one proposed
copy of the purchased product whenever a purchase occurs.
Conditional on the current virtual states and proposals,
averaging over purchase outcomes must recover the prescribed
joint law of independent copy updates. We characterize
coupling feasibility through Hall-type coverage conditions
and use a priority-rule decomposition to construct the
conditional sampler.

\subsection{Main Results and Contributions}
\label{subsec:contributions}

Our main contributions are a general revenue-extraction bound
under Markov chain choice and approximation guarantees for
consecutive-stay systems with arbitrary product capacities.

\begin{itemize}[leftmargin=10pt]

\item \textbf{A general revenue-extraction bound under Markov chain choice.}
We establish a general bound on the expected revenue from
optimizing over a random candidate set with independent
product inclusions, relative to the expected revenue of an
arbitrary randomized benchmark assortment
(Theorem~\ref{thm:general-static}).
We express the bound in terms of candidate-inclusion
probabilities and benchmark purchase probabilities normalized
by singleton purchase probabilities.
Our analysis does not require candidate-inclusion probabilities
to sum to at most one or match the benchmark's
product-inclusion probabilities.
We prove the bound using a continuation decomposition of
Markov paths and simultaneous independent thinning of the
candidate set.

The result applies beyond the consecutive-stay setting.
Applying it to our proposal mechanism, we show that optimizing
over products with at least one proposing copy recovers at
least half of each period's relaxation revenue in expectation,
for arbitrary copy multiplicities.

\item \textbf{Purchase--update coupling for dynamic implementation.}
We characterize when a purchase outcome can be coupled with
prescribed independent Bernoulli product updates so that every
purchased product receives an update.
Necessary and sufficient Hall-type coverage inequalities describe
the feasible purchase distributions.
Using the classical greedy representation of the resulting
coverage polymatroid \citep{edmonds2003submodular}, we construct
a priority-rule decomposition and a sampler conditional on the
observed purchase outcome.

We lift these product updates to individual proposed copies,
assigning every purchase to a copy selected for a virtual update.
Conditional on the current virtual states, proposals, and selected
assortment, averaging over purchase outcomes recovers the prescribed
joint law of independent copy updates.
Other updated copies undergo only virtual discards and consume
no physical capacity.
This construction preserves the relaxation's interval marginals
and independence across virtual copies at the beginning of each
booking period, while keeping virtual availability within actual
availability.

Under general Markov chain choice, in each nonterminal period,
the policy offers the selected assortment with probability
\(1-1/e\) and the empty assortment otherwise.
This gate ensures coupling feasibility.
Combining this implementation with static revenue extraction gives
a \({(1-1/e)}/{2}\)-approximation relative to both the relaxation
and the optimal nonanticipative online value, for arbitrary product
capacities (Theorem~\ref{thm:main-guarantees}(i)).

\item \textbf{An improved MNL guarantee and sharpness results.}
Under MNL, all Hall-type coupling conditions hold without gating,
so dynamic implementation introduces no additional revenue loss.
The resulting policy achieves a \({1}/{2}\)-approximation
relative to both benchmarks for arbitrary product capacities
(Theorem~\ref{thm:main-guarantees}(ii)).
Under the common fixed-type Bernoulli-arrival setting, this improves
the general MNL guarantee of \citet{huwu2026} from \(0.271\) to
\(1/2\), even with one copy per product.
When the outside option has zero attractiveness in every period,
the same policy attains a \(1-1/e\) guarantee relative to both
benchmarks for arbitrary copy multiplicities, matching their
guarantee for this case (Remark~\ref{rem:mnl-zero-outside}).

We establish separate sharpness of the \(1/2\) static extraction
factor and the \(1-1/e\) uniform scalar-gate factor.
The extraction factor is sharp even on one-period, unit-capacity
instances with uniquely determined optimal sales flows.
These results do not establish tightness of the combined Markov
guarantee.
Under MNL, one-period instances with one copy per product show
that the proposed policy's revenue ratio relative to the optimal
online value approaches \(1/2\), establishing tightness for this
policy.

\end{itemize}

Building on the interval-flow and proposal-discarding framework
of \citet{huwu2026}, we use a copy-expanded sales-based linear
program (SBLP) that upper-bounds the optimal online value and
provides target flows for independent copy proposals.
To handle multiple copies, the relaxation combines product-level
customer choice with maximal-free-interval marginals for individual
copies.
Aggregating interchangeable copies yields an equivalent relaxation
whose number of interval state and sales-flow variables is
independent of the product capacities.
A symmetric lift recovers feasible copy-level flows with the same
objective value.


\section{Model and Benchmark}\label{sec:model}

We introduce the consecutive-stay system with multiple copies per product, specify the customer-choice model, and define the online benchmark. Customers choose among products, while the platform assigns each purchase to a physical copy available throughout the requested service interval.

\subsection{Multi-Capacity System and Consecutive Stays}
\label{subsec:model-capacity}

We consider a booking system with \(T\) periods in which
the platform selects product assortments and assigns
purchases to physical copies. The objective is to maximize
total expected revenue over the booking horizon.

\textbf{Products and physical copies.} 
Let \(\cM=\{1,\ldots,M\}\) denote the set of customer-visible
products. Product \(i\) has \(C_i\in\mathbb Z_+\) interchangeable
physical copies, indexed by \(c\in[C_i]\). Each copy \((i,c)\)
can fulfill only product \(i\). Copies are distinguished solely
for capacity accounting and assignment. Regardless of \(C_i\),
product \(i\) appears at most once in an offered assortment.
Since at most \(T\) bookings can be accepted over the horizon,
replacing \(C_i\) by \(\min\{C_i,T\}\) leaves the optimal online
value unchanged. We therefore assume \(C_i\le T\) without loss
of generality.

\textbf{Service slots and booking periods.}
Let \(\cN=\{1,\ldots,N\}\) denote the ordered set of service
slots and \(\cT=\{1,\ldots,T\}\) the set of booking periods.
Once a service slot on a copy is reserved, it remains unavailable for the rest of the booking horizon. At the beginning of period
\(t\), let \(\mathcal F_{tic}\subseteq\cN\) denote the free
service slots on copy \((i,c)\). Initially, every copy is fully
available: \(\mathcal F_{1ic}=\cN\) for all \(i\in\cM\) and
\(c\in[C_i]\).

\textbf{Customer arrivals and requests.}
We adopt the fixed-type Bernoulli-arrival model of
\citet{huwu2026}.
Each booking period \(t\) has at most one potential customer,
who arrives with probability \(p_t\in(0,1]\), independently
across periods. The customer's request is for the consecutive
service interval \(I_t=[\ell_t,r_t]\subseteq\cN\), where
\(1\le\ell_t\le r_t\le N\).
The arrival probabilities, requested intervals, revenues, and customer-choice models described below are known for all periods before the booking horizon begins. Arrival and choice outcomes remain uncertain and are observed only as they occur.

\textbf{Feasibility and sequence of events.}
Each accepted booking is irrevocably assigned to a single
physical copy for its entire requested service interval. Thus, copy
\((i,c)\) can fulfill the period-\(t\) request if and only if
\(I_t\subseteq\mathcal F_{tic}\). Let
\(\mathcal C_t(i):=\{c\in[C_i]:I_t\subseteq\mathcal F_{tic}\}\)
denote the set of feasible copies of product \(i\), and let
\(\cA_t:=\{i\in\cM:\mathcal C_t(i)\ne\varnothing\}\)
denote the set of available products in period \(t\).
The events in each period occur in the following order:
\begin{enumerate}[label=(\roman*)]
    \item Before observing the arrival outcome, the platform
    selects an assortment \(S_t\subseteq\cA_t\).

    \item The arrival outcome is observed. If the customer
    arrives, the platform offers the preselected assortment
    \(S_t\).

    \item The arriving customer chooses a product \(i\in S_t\)
    or makes no purchase, according to the choice model
    described below.

    \item If product \(i\) is purchased, the platform assigns
    the purchase to a feasible copy \(c_t\in\mathcal C_t(i)\)
    and removes the requested interval from that copy's
    free-slot set:
    \(\mathcal F_{t+1,i,c_t}
    =\mathcal F_{t,i,c_t}\setminus I_t\).
    All other free-slot sets remain unchanged.
\end{enumerate}
If no customer arrives or the customer makes no purchase,
all free-slot sets remain unchanged.

\subsection{Markov Chain Choice}
\label{subsec:model-choice}

Conditional on an arrival in period \(t\), the customer chooses
from the preselected assortment \(S_t\). For any assortment
\(S\subseteq\cM\), let \(\phi_{ti}(S)\) denote the purchase
probability of product \(i\in S\), and let
\(\phi_{t0}(S):=1-\sum_{i\in S}\phi_{ti}(S)\) denote the
no-purchase probability. We set \(\phi_{ti}(S)=0\) for
\(i\in\cM\setminus S\). A purchase of product \(i\) generates
revenue \(r_{ti}\ge0\).

\textbf{Markov chain representation.}
For each period \(t\), the choice probabilities are generated
by a Markov chain on \(\cM\cup\{0\}\), where state \(0\)
represents no purchase and is absorbing. The chain has initial
distribution \(\lambda_t=(\lambda_{ti})_{i\in\cM\cup\{0\}}\)
and transition matrix
\(Q_t=(q_{tij})_{i,j\in\cM\cup\{0\}}\), with self-transitions
allowed. The underlying customer trajectories are independent
across periods and of the arrival process and the
platform's internal randomization. Given an assortment \(S\), the chain starts according
to \(\lambda_t\) and stops upon its first visit to \(S\cup\{0\}\).
If the initial state belongs to this set, it stops immediately.
If the stopping state is \(i\in S\), the customer purchases
product \(i\); if it is \(0\), the customer makes no purchase.

We assume that the submatrix of \(Q_t\) restricted to the product
states is transient. Consequently, the chain reaches
\(S\cup\{0\}\) almost surely for every \(S\subseteq\cM\),
including \(S=\varnothing\). The purchase probabilities are
therefore well defined by
\[
\phi_{ti}(S)
:=
\Prb_{\lambda_t,Q_t}
\left\{
i\text{ is the first state visited in }S\cup\{0\}
\right\},
\qquad i\in S,
\]
and the probabilities over \(S\cup\{0\}\) sum to one.

\textbf{Regularity and singleton purchase probabilities.}
The hitting-time representation implies the standard regularity
property \citep{blanchet2016markov}: for
\(S'\subseteq S\subseteq\cM\) and \(i\in S'\),
\(\phi_{ti}(S')\ge\phi_{ti}(S)\). In particular,
\begin{equation}
\phi_{ti}(S)
\le
a_{ti}:=\phi_{ti}(\{i\}),
\qquad i\in S.
\label{eq:singleton-upper}
\end{equation}
Thus, \(a_{ti}\) is the maximum purchase probability of
product \(i\) over all assortments containing it.

If \(a_{ti}=0\), nonnegativity and \eqref{eq:singleton-upper}
imply \(\phi_{ti}(S)=0\) for every \(S\subseteq\cM\).
Such a product is never purchased in period \(t\) and can be
omitted from that period's assortment. We therefore restrict
attention to products with \(a_{ti}>0\) in each period.

\textbf{MNL special case.}
The multinomial logit (MNL) model is a special case of Markov
chain choice \citep{blanchet2016markov}. When the outside option
has positive attractiveness, we normalize its attractiveness
to one. Under the classical MNL formulation
\citep{luce1959individual,mcfadden1974conditional},
\begin{equation}
\phi_{ti}(S)
=
\frac{u_{ti}}{1+\sum_{j\in S}u_{tj}},
\qquad i\in S,
\label{eq:mnl-choice}
\end{equation}
where \(u_{ti}>0\) is the normalized attractiveness of
product \(i\). Hence, \(a_{ti}=u_{ti}/(1+u_{ti})\). Remark~\ref{rem:mnl-zero-outside} treats the case of
zero outside-option attractiveness.

\subsection{Online Policies, Performance Benchmark, and Main Results}
\label{subsec:model-benchmark}

A feasible nonanticipative online policy may use internal
randomization to select its period-\(t\) assortment based on
the model primitives, the history through period \(t-1\),
and the current free-slot sets, without observing current
or future arrival and choice realizations. If a purchase occurs,
the policy assigns it to a feasible copy of the purchased product.
For an instance \(I\), let \(\Pi^{\rm on}(I)\) denote the set of
feasible nonanticipative online policies, and let \(V^\pi(I)\)
denote the expected revenue earned by policy
\(\pi\in\Pi^{\rm on}(I)\). Our performance benchmark is the
optimal online value,
\(V^{\rm on}(I):=\sup_{\pi\in\Pi^{\rm on}(I)}V^\pi(I)\).

A booking reserves a consecutive service interval on a single
copy, preventing that copy from serving later requests with
overlapping service intervals. Customer substitution makes
expected revenue depend nonlinearly on the offered assortment.
The platform must therefore coordinate assortment and assignment
decisions while tracking each copy's availability.
We establish the following approximation guarantees.

\begin{theorem}[Multi-capacity approximation guarantees]
\label{thm:main-guarantees}
Under the fixed-type Bernoulli-arrival model, there exist feasible nonanticipative policies with the following guarantees:
\begin{enumerate}[label=(\roman*),ref=(\roman*),leftmargin=*,itemsep=2pt]
\item \textit{Markov chain choice}. For every instance \(I\) under Markov chain choice, policy \(\pi_{\rm MC}\) satisfies
\begin{equation}\label{eq:markov-main}
V^{\pi_{\rm MC}}(I)
\ge
\frac{1-1/e}{2}V^{\rm on}(I)
\end{equation}

\item \textit{MNL choice}. For every instance \(I\) under MNL choice with positive outside-option attractiveness in every period, policy \(\pi_{\rm MNL}\) satisfies
\begin{equation}\label{eq:mnl-main}
V^{\pi_{\rm MNL}}(I)
\ge
\frac{1}{2}V^{\rm on}(I)
\end{equation}

\end{enumerate}
Both policies admit exact implementations in expected time polynomial in 
\(T,N,M\) and the size of the model input. The approximation factors are independent of the product capacities.
\end{theorem}




The policies \(\pi_{\rm MC}\) and \(\pi_{\rm MNL}\) are the
Markov and MNL versions of Algorithm~\ref{alg:general-policy} in Section \ref{sec:performance},
respectively; their construction and analysis follow below.
The next section introduces a relaxation that links
copy-level availability to product-level choice probabilities
and provides the target flows used to construct these policies.


\section{Sales-Based Fluid Relaxation}\label{sec:sblp}

We first formulate the relaxation and show that its optimal value
upper-bounds the optimal online value. We then derive an equivalent
aggregate formulation that can be solved efficiently.

\subsection{Copy-Expanded SBLP and Its Benchmark Bound}
\label{subsec:sblp-formulation}

The availability of a physical copy is described by its free-slot
set, which can be any subset of \(\cN\). Tracking its distribution
can therefore require up to \(2^N\) states per copy. Our relaxation
decomposes each free-slot set into maximal consecutive intervals
and tracks their marginal probabilities without retaining their
joint distribution. Related interval-based representations appear
in revenue-management and seat-allocation models with interval
demand \citep{rusmevichientong2023revenue,zhu2023assign} and form
the basis of the consecutive-stay relaxation of \citet{huwu2026}.

Recall that \(\mathcal F_{tic}\subseteq\cN\) denotes the free service
slots on copy \((i,c)\) at the beginning of period \(t\). For any
free-slot set \(\mathcal F\subseteq\cN\), an interval
\([a,b]\subseteq\mathcal F\) is a \emph{maximal free interval} if
it cannot be extended in either direction while remaining in
\(\mathcal F\). Its endpoints must therefore satisfy both conditions:
\(a=1\) or \(a-1\notin\mathcal F\), and \(b=N\) or
\(b+1\notin\mathcal F\). We denote this relationship by
\([a,b]\sim\mathcal F\). For example, the maximal free intervals
of \(\mathcal F=\{1,2,3,6,7\}\) are \([1,3]\) and \([6,7]\).
The interval \([1,2]\) is free but not maximal because slot \(3\)
is also free. Every free-slot set has a unique decomposition into
disjoint maximal free intervals.

A copy can fulfill the request \(I_t=[\ell_t,r_t]\) if and only if
one of its maximal free intervals contains \(I_t\). Since these
intervals are disjoint, the containing interval is unique whenever
it exists.
Fulfilling the request from \([a,b]\) removes \(I_t\), leaving
the left residual interval \([a,\ell_t-1]\) when \(a<\ell_t\)
and the right residual interval \([r_t+1,b]\) when \(r_t<b\).
Thus, each sale removes one maximal free interval and creates
at most two smaller ones, leaving all other maximal free intervals
unchanged.

Let \(\cQ:=\{[a,b]:1\le a\le b\le N\}\) denote the set of all
\(N(N+1)/2\) service intervals. For each \(t\in\cT\), \(i\in\cM\),
\(c\in[C_i]\), and \([a,b]\in\cQ\), we introduce nonnegative
variables with the following intended probabilistic interpretations:
\begin{itemize}
\item
\(x_{tic}([a,b])\) is the probability that \([a,b]\) is a maximal
free interval of copy \((i,c)\) at the beginning of period \(t\).

\item
\(y_{tic}([a,b])\) is the joint probability that \([a,b]\) is a
maximal free interval of copy \((i,c)\) at the beginning of
period \(t\), the customer arrives and purchases product \(i\),
and the purchase is assigned to this copy and fulfilled from
\([a,b]\).
\end{itemize}

We remark that a copy may have several maximal free intervals simultaneously, so
the events represented by the \(x\)-variables are not mutually
exclusive. Under this interpretation,
\(\sum_{[a,b]\in\cQ}x_{tic}([a,b])\) equals the expected number
of maximal free intervals on copy \((i,c)\) at the beginning of
period \(t\) and need not equal one. 

Following the sales-based linear programming framework of
\citet{gallego2015general}, we combine copy-level interval flows
with product-level choice variables. For each period \(t\) and
assortment \(S\subseteq\cM\), let \(\theta_{tS}\ge0\) denote the
mixture weight on \(S\), and let \(z_{ti}\) denote the resulting
purchase probability of product \(i\), conditional on an arrival.

The copy-expanded sales-based linear program (SBLP) is
\begin{equation}
\SBLP^C(I)
=
\max_{\substack{x\ge0,\;y\ge0\\ \theta\ge0,\;z\ge0}}
\quad
\sum_{t\in\cT}
\sum_{i\in\cM}
\sum_{c=1}^{C_i}
\sum_{[a,b]\in\cQ}
r_{ti}y_{tic}([a,b])
\label{eq:sblp-objective}
\end{equation}
subject to
\begin{align}
\sum_{S\subseteq\cM}\theta_{tS}
&=
1,
&& t\in\cT,
\label{eq:sblp-theta}\\
z_{ti}
&=
\sum_{S\ni i}\theta_{tS}\phi_{ti}(S),
&& t\in\cT,\ i\in\cM,
\label{eq:sblp-choice}\\
\sum_{c=1}^{C_i}
\sum_{[a,b]\in\cQ}
y_{tic}([a,b])
&=
p_tz_{ti},
&& t\in\cT,\ i\in\cM,
\label{eq:sblp-sales-link}\\
y_{tic}([a,b])
&\le
p_ta_{ti}x_{tic}([a,b])
\one\{I_t\subseteq[a,b]\},
&&
\substack{
t\in\cT,\ i\in\cM,\\
c\in[C_i],\ [a,b]\in\cQ,
}
\label{eq:sblp-state-cap}\\
x_{tic}([a,b])
&=
x_{t-1,i,c}([a,b])
-y_{t-1,i,c}([a,b])
\notag\\
&\quad+
\sum_{\substack{
b'\ge r_{t-1}\\
{[a,b']\in\cQ}
}}
\one\{b+1=\ell_{t-1}\}
y_{t-1,i,c}([a,b'])
\notag\\
&\quad+
\sum_{\substack{
a'\le\ell_{t-1}\\
{[a',b]\in\cQ}
}}
\one\{r_{t-1}=a-1\}
y_{t-1,i,c}([a',b]),
&&
\substack{
t=2,\ldots,T,\ i\in\cM,\\
c\in[C_i],\ [a,b]\in\cQ,
}
\label{eq:sblp-balance}\\
x_{1ic}([a,b])
&=
\one\{[a,b]=[1,N]\},
&&
\substack{
i\in\cM,\ c\in[C_i],\\
{[a,b]\in\cQ.}
}
\label{eq:sblp-boundary}
\end{align}

Equations~\eqref{eq:sblp-theta}--\eqref{eq:sblp-sales-link}
connect assortment choices to sales. The weights \(\theta_{tS}\)
form a distribution over assortments, and \(z_{ti}\) is the
corresponding purchase probability conditional on an arrival.
The unconditional sales probability \(p_tz_{ti}\) is then allocated
across copies and intervals through the \(y\)-variables.
By \eqref{eq:sblp-sales-link}, the objective can equivalently be
written as
\(\sum_{t\in\cT}p_t\sum_{i\in\cM}r_{ti}z_{ti}\).

Constraint~\eqref{eq:sblp-state-cap} bounds the sales probability
associated with each copy and interval. Conditional on the current
free-slot sets and the preselected assortment, the arrival
probability remains \(p_t\). Given an arrival, the purchase
probability of product \(i\) is at most \(a_{ti}\) by
\eqref{eq:singleton-upper}. Fulfillment from \([a,b]\) further
requires that it be a maximal free interval of copy \((i,c)\)
and contain \(I_t\). Since the maximal-interval event has
probability \(x_{tic}([a,b])\), these observations yield
\eqref{eq:sblp-state-cap}.

Equations~\eqref{eq:sblp-balance}--\eqref{eq:sblp-boundary}
describe the evolution and initialization of the interval marginals.
In \eqref{eq:sblp-balance}, the first two terms give the mass of
\([a,b]\) that persists after subtracting sales fulfilled from
that interval. The first inflow term records sales from a parent
interval \([a,b']\) that create \([a,b]\) as the left residual,
requiring \(b=\ell_{t-1}-1\). The second records sales from a
parent interval \([a',b]\) that create \([a,b]\) as the right
residual, requiring \(a=r_{t-1}+1\).
The initial condition \eqref{eq:sblp-boundary} reflects that every
copy is initially free over \([1,N]\).

Using the probabilistic interpretations above, every feasible
nonanticipative online policy \(\pi\) induces a feasible SBLP
solution with objective value \(V^\pi(I)\). Taking the supremum
over such policies yields the following upper bound.

\begin{proposition}[SBLP upper bound]\label{prop:sblp-upper}
For every instance \(I\),
\begin{equation}
V^{\rm on}(I)
\le
\SBLP^C(I).
\label{eq:sblp-upper}
\end{equation}
\end{proposition}

\subsection{Exact Symmetry Reduction and Tractability}
\label{subsec:sblp-symmetry}

Copies of the same product have identical initial conditions
and interval-flow constraints. We therefore aggregate their
interval state and sales-flow variables to reduce the size
of the relaxation. Define
\begin{equation}
X_{ti}([a,b])
:=
\sum_{c=1}^{C_i}x_{tic}([a,b]),
\qquad
Y_{ti}([a,b])
:=
\sum_{c=1}^{C_i}y_{tic}([a,b]).
\label{eq:aggregate-xy}
\end{equation}
The choice constraints
\eqref{eq:sblp-theta}--\eqref{eq:sblp-choice} remain unchanged,
since they involve only product-level variables. Substituting
\eqref{eq:aggregate-xy} into the objective
\eqref{eq:sblp-objective} gives
\(\sum_{t\in\cT}\sum_{i\in\cM}\sum_{[a,b]\in\cQ}
r_{ti}Y_{ti}([a,b])\).

The sales-linking constraint \eqref{eq:sblp-sales-link} becomes
\[
\sum_{[a,b]\in\cQ}Y_{ti}([a,b])
=
p_tz_{ti},
\qquad t\in\cT,\ i\in\cM.
\]
Summing the interval feasibility bounds
\eqref{eq:sblp-state-cap} over \(c\in[C_i]\) yields
\begin{equation}
0
\le
Y_{ti}([a,b])
\le
p_ta_{ti}X_{ti}([a,b])
\one\{I_t\subseteq[a,b]\},
\label{eq:aggregate-state-cap}
\end{equation}
for every \(t\in\cT\), \(i\in\cM\), and \([a,b]\in\cQ\).
Similarly, summing the balance equations
\eqref{eq:sblp-balance} over \(c\in[C_i]\) gives, for
\(t=2,\ldots,T\),
\begin{align}
X_{ti}([a,b])
={}&
X_{t-1,i}([a,b])
-
Y_{t-1,i}([a,b])
+
\sum_{\substack{
b'\ge r_{t-1}\\
{[a,b']\in\cQ}
}}
\one\{b+1=\ell_{t-1}\}
Y_{t-1,i}([a,b'])
\notag\\
&+
\sum_{\substack{
a'\le\ell_{t-1}\\
{[a',b]\in\cQ}
}}
\one\{r_{t-1}=a-1\}
Y_{t-1,i}([a',b]).
\label{eq:aggregate-balance}
\end{align}
Finally, summing the initial conditions
\eqref{eq:sblp-boundary} over copies gives
\(X_{1i}([a,b])=C_i\one\{[a,b]=[1,N]\}\).
All variables \(X,Y,\theta,z\) are nonnegative.

Aggregation maps every feasible copy-expanded solution to
a feasible aggregate solution with the same objective value.
Conversely, every feasible aggregate solution admits the
symmetric lift
\[
x_{tic}([a,b])
=
\frac{X_{ti}([a,b])}{C_i},
\qquad
y_{tic}([a,b])
=
\frac{Y_{ti}([a,b])}{C_i},
\qquad c\in[C_i],
\]
with \(\theta\) and \(z\) unchanged. Each copy's interval
constraints and initial conditions follow by dividing their
aggregate counterparts by \(C_i\). The lifted sales flows
sum to the original aggregate flows, preserving both the
sales-linking constraints and the objective value.
This establishes the following result.

\begin{proposition}[Exact symmetry reduction]
\label{prop:sblp-symmetry}
The aggregate and copy-expanded SBLP formulations have the same
optimal value. The aggregate formulation contains \(O(TMN^2)\)
interval state and sales-flow variables, independently of the
product capacities \(C_i\).
\end{proposition}

Although aggregation reduces the interval variables, the
assortment-mixture variables \(\theta_{tS}\) remain exponential
in number. After eliminating \(z\) using
\eqref{eq:sblp-choice}, the unrestricted dual multipliers
\(\alpha_t\) and \(\beta_{ti}\) for \eqref{eq:sblp-theta}
and the aggregate form of \eqref{eq:sblp-sales-link},
respectively, yield the constraints
\[
\alpha_t
\ge
p_t\sum_{i\in S}\beta_{ti}\phi_{ti}(S),
\qquad t\in\cT,\ S\subseteq\cM.
\]
Checking these constraints amounts to maximizing their
right-hand side for each period. By regularity, it suffices
to solve the following problem for eligible products
\(A\subseteq\cM\) and weights \(w\in\R_+^M\):
\begin{equation}
\max_{S\subseteq A}
\sum_{i\in S}w_i\phi_{ti}(S).
\label{eq:weighted-static}
\end{equation}
For dual separation, take
\(A=\{i\in\cM:\beta_{ti}>0\}\) and
\(w_i=\max\{\beta_{ti},0\}\).
Products outside \(A\) remain transition states but cannot
be offered.
This Markov assortment problem can be solved in polynomial
time through an optimal-stopping formulation
\citep{blanchet2016markov}. Together with direct checks of the remaining dual constraints,
this oracle allows the constructive ellipsoid method to solve
the aggregate SBLP and recover an optimal primal solution
without enumerating all assortments
\citep{grotschel1981ellipsoid}.



Let \(C_{\rm tot}:=\sum_{i\in\cM}C_i\) denote the total number
of physical copies. Accounting also for the construction
of the symmetric copy-level lift yields the following result.

\begin{proposition}[Polynomial solvability of the aggregate relaxation]
\label{prop:sblp-polytime-general}
An optimal aggregate SBLP solution and its symmetric
copy-level lift can be computed in time polynomial in
\(T,N,M,C_{\rm tot}\) and the size of the model input.
The assortment mixture can be represented using at most
\(M+1\) positive weights per period.
\end{proposition}

Through the symmetric lift, we assign identical target flows to copies of the same product. Explicitly representing these flows requires \(O(TN^2C_{\rm tot})\) entries. In the online policy, we also track actual and virtual availability separately for each copy. Thus, although aggregation makes the interval-variable count independent of product capacities, constructing the copy-level flows and implementing the policy still depend on \(C_{\rm tot}\). Since \(C_i\le T\), we have \(C_{\rm tot}\le MT\), so this dependence remains polynomial in the model dimensions. We use these target flows to generate proposals in Section~\ref{sec:proposal-revenue}. Before presenting that construction, we establish a general revenue-extraction bound for independent candidate sets.


\section{Revenue Extraction under Markov Chain Choice}
\label{sec:general-static}

We establish a static revenue bound for independent candidate sets relative to an arbitrary randomized-assortment benchmark. The result is independent of the consecutive-stay setting and will later be applied to the candidate sets generated by our copy-proposal mechanism.

\textbf{Choice Model and Benchmark.}
Consider the Markov chain choice model defined in
Subsection~\ref{subsec:model-choice}, with the time index suppressed.
Let \(r_i\ge0\) denote the revenue of product \(i\), and recall
its singleton purchase probability \(a_i:=\phi_i(\{i\})\).
Write \(R(S):=\sum_{i\in S}r_i\phi_i(S)\) for the expected revenue
from offering assortment \(S\), and define
\(f(A):=\max_{S\subseteq A}R(S)\) for a candidate set
\(A\subseteq\cM\).

Let \(B\) be an arbitrary random benchmark assortment with
\(\Prb(B=S)=\theta_S\); its product-inclusion indicators need
not be independent. Define the benchmark purchase
probabilities as
\(z_i:=\E[\phi_i(B)]=\sum_{S\ni i}\theta_S\phi_i(S)\)
and its expected revenue as
\(W:=\sum_{i\in\cM}r_i z_i\).
We assume \(W>0\); otherwise, the revenue lower bound
follows from nonnegativity.

\textbf{Independent Candidate Sets.}
Let \(P\subseteq\cM\) have mutually independent product-inclusion
indicators, with \(p_i:=\Prb(i\in P)\).
Define the normalized benchmark target as \(\rho_i:=z_i/a_i\)
when \(a_i>0\), and set \(\rho_i:=0\) when \(a_i=0\).
Regularity gives \(0\le z_i\le a_i\), so the latter convention
applies only when \(z_i=0\), and \(0\le\rho_i\le1\).
For a realized candidate set \(P\), \(f(P)\) is the maximum expected revenue obtainable by offering an assortment contained in \(P\).

\begin{theorem}[Revenue extraction under Markov choice]
\label{thm:general-static}
For any independent candidate set \(P\) with
\(\Prb(i\in P)=p_i\),
\begin{equation}
\E[f(P)]
\ge
\left(
\min_{i:z_i>0}
\frac{p_i}{p_i+\rho_i(1-p_i)}
\right)
\sum_{i\in\cM}r_i z_i.
\label{eq:general-static}
\end{equation}
In particular, if
\(p_i\ge\rho_i/(1+\rho_i)\) for every product with \(z_i>0\),
then \(\E[f(P)]\ge\frac12\sum_{i\in\cM}r_i z_i\).
\end{theorem}


The theorem allows arbitrary inclusion probabilities \(p_i\in[0,1]\), without requiring their sum to be at most one. Its comparison differs from a standard correlation-gap comparison, which fixes a set function and compares independent and correlated distributions with the same inclusion marginals \citep{agrawal2010correlation}. Here, the benchmark is \(\mathbb E[R(B)]\), and \(z_i=\mathbb E[\phi_i(B)]\) records a purchase probability rather than an inclusion probability. The candidate probabilities \(p_i\) need not equal either \(z_i\) or \(\Pr(i\in B)\). The bound instead connects independent product availability to benchmark purchases through \(\rho_i=z_i/a_i\).

\textbf{Relation to \citet{huwu2026}.}
In their analysis of the MNL setting with positive outside-option attractiveness, \citet{huwu2026} study the corresponding static revenue-extraction problem. In our notation, their proposal mechanism yields \(p_i=\rho_i=z_i/a_i\), and their analysis establishes \(\mathbb E[f(P)]\ge0.271W\).
Their proof combines a convexity bound for offering all
low-attractiveness proposals with a tail-probability bound
for the best high-attractiveness singleton, optimizing
the threshold between the groups.


Theorem~\ref{thm:general-static} gives a unified analysis
through a Markov-path decomposition and independent thinning.
For the same MNL proposal probabilities \(p_i=\rho_i\), it yields
\[
\E[f(P)]
\ge
\left(\min_{i:z_i>0}\frac{1}{2-\rho_i}\right)W
\ge \frac12 W.
\]
Thus, for the same proposal probabilities, our analysis strengthens the single-period revenue guarantee from \(0.271W\) to at least \(W/2\).
More generally, the theorem applies to Markov chain choice
and allows candidate-inclusion probabilities that differ
from the normalized benchmark targets.


\subsection{Proof of Theorem~\ref{thm:general-static}}
\label{sec:proof-static-general}

The proof has two steps. First, a Markov-path decomposition relates singleton purchase probabilities to the benchmark purchase probabilities. Second, we independently thin the candidate set so that every product attains the same fraction of its benchmark purchase probability. The decomposition shows that a product retained whenever available meets its target whenever all other products remain below theirs. Regularity then allows us to attain all targets simultaneously. Because the thinned assortment is contained in \(P\), optimizing over \(P\) earns at least its expected revenue.

\subsubsection{Markov-Path Decomposition and the Benchmark Bound}
\label{subsec:continuation-class}

Let \(Z_0,Z_1,\ldots\) follow the underlying Markov chain
without stopping at offered products. Define the hitting times
\(T_i:=\inf\{n\ge0:Z_n=i\}\) and
\(T_0:=\inf\{n\ge0:Z_n=0\}\).
For a chain started from the initial distribution \(\lambda\),
\(a_i=\Prb_\lambda(T_i<T_0)\).
For a chain started from product \(j\), write
\(h_{ji}:=\Prb_j(T_i<T_0)\).
These continuation probabilities satisfy \(0\le h_{ji}\le1\)
and \(h_{ii}=1\), since hitting times include time zero.

The following decomposition separates paths according to
the first offered product they reach, then follows their
continuation toward product \(i\).

\begin{proposition}[Markov-path decomposition]
\label{prop:markov-continuation}
For every product \(i\) and assortment \(S\subseteq\cM\),
\[
\sum_{j\in S}h_{ji}\phi_j(S)\le a_i.
\]
Equality holds whenever \(i\in S\).
\end{proposition}

\proof{Proof.}
Let \(\tau_S:=\inf\{n\ge0:Z_n\in S\cup\{0\}\}\).
Transience of the product states makes \(\tau_S\) finite
almost surely. For each \(j\in S\), define
\[
E_j:=
\left\{
Z_{\tau_S}=j,\quad
Z_n=i\text{ for some }n\text{ with }\tau_S\le n<T_0
\right\}.
\]
The strong Markov property gives
\(\Prb_\lambda(E_j)=\phi_j(S)h_{ji}\).
The events \(E_j\) are mutually disjoint and contained in
\(\{T_i<T_0\}\). Hence,
\[
\sum_{j\in S}h_{ji}\phi_j(S)
=
\Prb_\lambda\!\left(\bigcup_{j\in S}E_j\right)
\le
\Prb_\lambda(T_i<T_0)
=
a_i.
\]
If \(i\in S\), every path with \(T_i<T_0\) satisfies
\(\tau_S\le T_i\) and first reaches some \(j\in S\).
It therefore belongs to exactly one event \(E_j\), giving
\(\bigcup_{j\in S}E_j=\{T_i<T_0\}\) and proving equality.
For \(S=\varnothing\), the inequality is immediate.
\hfill\(\blacksquare\)
\endproof

Averaging this decomposition over the benchmark assortment
gives a bound expressed directly in terms of its purchase
probabilities.

\begin{lemma}[Benchmark continuation bound]
\label{lem:general-continuation-budget}
For every product \(i\),
\[
\sum_{j\ne i}h_{ji}z_j\le a_i-z_i.
\]
\end{lemma}

\proof{Proof.}
Apply Proposition~\ref{prop:markov-continuation} to each
realization of \(B\) and take expectations. Since
\(\phi_j(B)=0\) when \(j\notin B\) and \(h_{ii}=1\),
\[
z_i+\sum_{j\ne i}h_{ji}z_j
=
\sum_{j\in\cM}h_{ji}\E[\phi_j(B)]
\le a_i.
\]
Rearranging proves the claim. No independence assumption
on the benchmark assortment is needed.
\hfill\(\blacksquare\)
\endproof

\subsubsection{Independent Thinning to Match the Purchase Targets}
\label{subsec:balanced-attenuation-general}

Let \(I^+:=\{i\in\cM:z_i>0\}\).
If \(p_i=0\) for some \(i\in I^+\), the coefficient in
\eqref{eq:general-static} is zero and the theorem is immediate.
We therefore assume \(p_i>0\) for every \(i\in I^+\), and set
\[
\vartheta:=
\min_{i\in I^+}
\frac{p_i}{p_i+\rho_i(1-p_i)}.
\]
Since \(0<\rho_i\le1\) on \(I^+\), we have
\(0<\vartheta\le1\).
Our goal is to construct a random subset of \(P\) with
purchase probabilities exactly \((\vartheta z_i)_{i\in\cM}\).

For a retention vector \(\beta\in[0,1]^M\), independently
retain each product \(i\in P\) with probability \(\beta_i\),
using draws independent of \(P\).
Denote the resulting assortment by \(Q^\beta\subseteq P\),
and write \(\xi_i(\beta):=\E[\phi_i(Q^\beta)]\).
Its inclusion indicators are independent, with probabilities
\(p_i\beta_i\).

Each \(\xi_i\) is continuous in \(\beta\), since it is a
finite sum of purchase probabilities weighted by polynomials
in the retention probabilities.
Moreover, writing \(Q^\beta_{-i}:=Q^\beta\setminus\{i\}\),
independence gives
\[
\xi_i(\beta)
=
p_i\beta_i\,
\E\!\left[\phi_i(\{i\}\cup Q^\beta_{-i})\right].
\]
The expectation does not depend on \(\beta_i\), so
\(\xi_i\) is linear and nondecreasing in its own coordinate.
Under common random draws, increasing another coordinate
\(\beta_j\) can only add a competing product; regularity
therefore implies that it cannot increase \(\xi_i\).

The next lemma supplies the key boundary condition:
if all other purchase probabilities do not exceed their targets,
fully retaining product \(i\) suffices to reach its target.

\begin{lemma}[Purchase target at full retention]
\label{lem:general-one-coordinate}
Fix \(i\in I^+\).
If \(\beta_i=1\) and \(\xi_j(\beta)\le\vartheta z_j\)
for every \(j\ne i\), then
\(\xi_i(\beta)\ge\vartheta z_i\).
\end{lemma}

\proof{Proof.}
Write \(Q:=Q^\beta\) and \(Q_{-i}:=Q\setminus\{i\}\).
Since \(\beta_i=1\), product \(i\) belongs to \(Q\) with
probability \(p_i\), independently of \(Q_{-i}\).
For \(j\ne i\), define
\[
b_j^+:=\E[\phi_j(\{i\}\cup Q_{-i})],
\qquad
b_j^-:=\E[\phi_j(Q_{-i})].
\]
Independence gives \(\xi_j(\beta)=p_i b_j^++(1-p_i)b_j^-\) for every \(j\ne i\). Since regularity implies \(b_j^+\le b_j^-\), we obtain \(b_j^+\le\xi_j(\beta)\le\vartheta z_j\).
Every realization of \(\{i\}\cup Q_{-i}\) contains \(i\).
The equality case of
Proposition~\ref{prop:markov-continuation}, followed by
Lemma~\ref{lem:general-continuation-budget}, therefore gives
\[
\begin{aligned}
\xi_i(\beta)
&=p_i\E[\phi_i(\{i\}\cup Q_{-i})]
 =p_i\left[a_i-\sum_{j\ne i}h_{ji}b_j^+\right]\\
&\ge p_i\left[a_i-\vartheta\sum_{j\ne i}h_{ji}z_j\right]
 \ge p_i\bigl[a_i-\vartheta(a_i-z_i)\bigr]\\
&=\vartheta z_i
  +a_i\bigl[p_i-\vartheta\{p_i+\rho_i(1-p_i)\}\bigr]
 \ge\vartheta z_i,
\end{aligned}
\]
where the last inequality follows from the definition
of \(\vartheta\).
\hfill\(\blacksquare\)
\endproof

We now attain all targets simultaneously. Among retention vectors that keep every purchase probability at most its target, choose one maximizing total retention. If a product remains below its target, the preceding lemma implies that its retention probability is less than one. We can then increase this probability slightly without violating any target constraint, contradicting maximality.

\begin{lemma}[Simultaneous balanced attenuation]
\label{lem:general-balanced-attenuation}
There exists \(\beta^\star\in[0,1]^M\), with
\(\beta_i^\star=0\) for \(i\notin I^+\), such that
\[
\xi_i(\beta^\star)=\vartheta z_i,
\qquad i\in\cM.
\]
\end{lemma}

\proof{Proof.}
Consider
\[
\mathcal D:=
\left\{
\beta\in[0,1]^M:
\begin{array}{ll}
\beta_i=0 & \text{for }i\notin I^+,\\
\xi_i(\beta)\le\vartheta z_i & \text{for }i\in\cM
\end{array}
\right\}.
\]
The zero vector belongs to \(\mathcal D\), and continuity
of the functions \(\xi_i\) makes \(\mathcal D\) closed.
It is therefore a nonempty compact subset of \([0,1]^M\).
Choose \(\beta^\star\in\mathcal D\) maximizing
\(\sum_{i\in I^+}\beta_i\).

Suppose that \(\xi_i(\beta^\star)<\vartheta z_i\)
for some \(i\in I^+\).
Lemma~\ref{lem:general-one-coordinate} implies
\(\beta_i^\star<1\), since all other purchase probabilities
satisfy their target bounds.
By continuity, a sufficiently small increase in
\(\beta_i^\star\) preserves its own constraint.
This change cannot increase any \(\xi_j\) with \(j\ne i\),
so the resulting vector remains in \(\mathcal D\).
Its total retention is strictly larger, contradicting
the choice of \(\beta^\star\).

Thus, every target on \(I^+\) is attained.
For \(i\notin I^+\), the restriction \(\beta_i^\star=0\)
gives \(\xi_i(\beta^\star)=0=\vartheta z_i\), completing
the proof.
\hfill\(\blacksquare\)
\endproof

The balanced assortment now gives the desired revenue
comparison directly.

\proof{Proof of Theorem~\ref{thm:general-static}.}
The zero-bound cases were addressed above.
Otherwise, take \(\beta^\star\) from
Lemma~\ref{lem:general-balanced-attenuation}.
For every realization of \(P\) and the thinning draws,
\(Q^{\beta^\star}\subseteq P\), so
\(f(P)\ge R(Q^{\beta^\star})\).
Taking expectations gives
\(
\E[f(P)]\ge\E[R(Q^{\beta^\star})]
=\sum_{i\in\cM}r_i\xi_i(\beta^\star)
=\vartheta\sum_{i\in\cM}r_i z_i.
\)
Substituting the definition of \(\vartheta\) proves
\eqref{eq:general-static}.
For the final assertion, whenever \(z_i>0\),
\[
p_i\ge\frac{\rho_i}{1+\rho_i}
\quad\Longleftrightarrow\quad
p_i\ge\rho_i(1-p_i)
\quad\Longleftrightarrow\quad
\frac{p_i}{p_i+\rho_i(1-p_i)}\ge\frac12.
\]
Thus, the stated inclusion-probability condition gives
\(\vartheta\ge1/2\) and the half-revenue guarantee.
\hfill\(\blacksquare\)
\endproof

The thinning is solely a proof device.
The policy optimizes directly over \(P\) and does not need
to compute \(\beta^\star\).

\subsection{Sharpness}
The uniform factor \(1/2\) cannot be improved even under MNL choice with equal revenues and a deterministic benchmark assortment. This remains true when each candidate-inclusion probability equals its normalized benchmark target, \(p_i=\rho_i\).

\begin{proposition}[Sharpness of the uniform Markov extraction bound]
\label{prop:markov-extraction-sharpness}
For every \(\varepsilon>0\), there exist an MNL choice model,
equal positive product revenues, and a deterministic benchmark
assortment \(B\) with \(W=\sum_i r_i z_i>0\) such that an
independent candidate set \(P\) with inclusion probabilities
\(p_i=\rho_i=z_i/a_i\) satisfies
\(
{\E[f(P)]}/{W}
\le
{1}/{2}+\varepsilon.
\)
\end{proposition}


\section{Independent Proposals and Revenue Extraction}
\label{sec:proposal-revenue}

We now return to the consecutive-stay system. We first use the
SBLP target flows to generate independent copy proposals and derive
the probability that each product is visible
(Subsection~\ref{sec:proposal-generation}). We then apply
Theorem~\ref{thm:general-static} to bound the expected revenue from
the visible products relative to the SBLP benchmark
(Subsection~\ref{sec:static-general}). Finally, we specify the target
updates that reproduce the SBLP interval flows
(Subsection~\ref{subsec:sblp-realizability}).

\subsection{Proposal Generation from the Fluid Solution}
\label{sec:proposal-generation}

We build on the virtual-state construction of \citet{huwu2026}, which uses interval-flow variables to generate independent copy proposals. In their choice-based model, each physical copy is a customer-facing alternative; in our model, several interchangeable copies may belong to the same product. Treating these copies as separate alternatives without adjusting the choice model can change product purchase probabilities.
We therefore make a product eligible if any copy proposes,
retaining copy identities for assignment and updates.

Fix a feasible copy-level SBLP solution. To generate proposals,
we maintain a \emph{virtual free-slot state}
\(\widehat{\mathcal F}_{tic}\subseteq\cN\) for each copy
\((i,c)\). For now, suppose that these states are mutually
independent across copies at the beginning of each period and
satisfy
\begin{equation}\label{eq:virtual-interval-marginals}
\Prb\bigl([a,b]\sim\widehat{\mathcal F}_{tic}\bigr)
=x_{tic}([a,b]),
\qquad [a,b]\in\cQ.
\end{equation}
We use these states to define the proposal rule.
Subsection~\ref{subsec:sblp-realizability} constructs updates
preserving these properties.
Section~\ref{sec:performance} verifies them for the implemented policy.

\textbf{Copy proposals.}
In period \(t\), a copy can propose only if its virtual state
contains \(I_t\). When this condition holds, there is a unique
maximal free interval containing \(I_t\), which determines the
proposal probability. Suppose that \([a,b]\) is the maximal free
interval of \(\widehat{\mathcal F}_{tic}\) containing \(I_t\).
Conditional on this event, copy \((i,c)\) proposes with probability
\begin{equation}
\pi_{tic}([a,b])
:=\frac{y_{tic}([a,b])}{p_ta_{ti}x_{tic}([a,b])}.
\label{eq:copy-proposal}
\end{equation}
We define the ratio as zero when \(x_{tic}([a,b])=0\), in which
case \eqref{eq:sblp-state-cap} also gives \(y_{tic}([a,b])=0\).
We also set \(\pi_{tic}([a,b])=0\) for intervals that do not
contain \(I_t\). If the virtual state cannot accommodate \(I_t\),
the copy does not propose. Constraint~\eqref{eq:sblp-state-cap}
ensures that \(0\le\pi_{tic}([a,b])\le1\).

We use random draws independent across copies and of the current
virtual states. With one copy per product under MNL,
\eqref{eq:copy-proposal} recovers the rule of \citet{huwu2026}.
Here, \(a_{ti}\) is the singleton purchase probability under
general Markov choice. Normalization by \(p_ta_{ti}\) enables
the updates in Subsection~\ref{subsec:sblp-realizability} to
reproduce the SBLP interval flows \(y_{tic}([a,b])\).

\textbf{Proposal marginals and product visibility.}
Although a copy may have several maximal free intervals, at most
one can contain \(I_t\). The corresponding events are therefore
mutually exclusive. Define the marginal proposal probability
of copy \((i,c)\) and the total proposal mass of product \(i\) by
\begin{align}
\rho_{tic}
:=\sum_{[a,b]\in\cQ}
x_{tic}([a,b])\pi_{tic}([a,b])
=\frac{1}{p_ta_{ti}}
\sum_{[a,b]\in\cQ}y_{tic}([a,b])\quad \textrm{and} \quad 
\rho_{ti}
:=\sum_{c=1}^{C_i}\rho_{tic}.
\label{eq:rho-copy}
\end{align}
Thus, \(\rho_{tic}\) is the probability that copy \((i,c)\)
proposes, and \(\rho_{ti}\) is the expected number of proposing
copies of product \(i\).

We call product \(i\) \emph{visible} if at least one of its
copies proposes, and write \(\widetilde\rho_{ti}\) for its
visibility probability.
A visible product is eligible for the subsequent assortment
decision and appears at most once in the offered assortment.
Because several copies may propose simultaneously,
\(\widetilde\rho_{ti}\) need not equal the total proposal mass
\(\rho_{ti}\).
The following lemma connects these quantities to the SBLP
sales targets.

\begin{lemma}[Proposal marginals and visibility]
\label{lem:aggregate-identities}
Suppose that, at the beginning of period \(t\), the virtual free-slot states are mutually independent across copies and satisfy the interval-marginal identities in \eqref{eq:virtual-interval-marginals}.
Under the proposal rule above, the following statements hold.
\begin{enumerate}[label=(\roman*),ref=(\roman*),leftmargin=*,itemsep=2pt]

\item \textit{Sales targets and normalization.}
For every product \(i\) and copy \(c\in[C_i]\),
\[
z_{ti}=a_{ti}\rho_{ti},
\qquad
0\le\rho_{tic}\le\rho_{ti}\le1.
\]

\item \textit{Independent proposals and visibility.}
The copy-proposal indicators are mutually independent, with
copy \((i,c)\) proposing with probability \(\rho_{tic}\).
The visible-product indicators are also mutually independent
across products, with
\begin{equation}
\widetilde\rho_{ti}
=1-\prod_{c=1}^{C_i}(1-\rho_{tic}).
\label{eq:visible-rho}
\end{equation}

\item \label{item:proposal-capacity-aggregation}\textit{Capacity aggregation.}
For every product \(i\),
\[
\widetilde\rho_{ti}
\ge1-e^{-\rho_{ti}}
\ge\frac{\rho_{ti}}{1+\rho_{ti}}
\ge\frac{\rho_{ti}}{2}.
\]
\end{enumerate}
\end{lemma}

Under the symmetric lift, all copies of product \(i\) have
the same marginal proposal probability,
\(\rho_{tic}=\rho_{ti}/C_i\). Hence,
\eqref{eq:visible-rho} simplifies to
\[
\widetilde\rho_{ti}
=1-\left(1-\frac{\rho_{ti}}{C_i}\right)^{C_i}.
\]
With one copy per product, visibility equals the normalized
sales target: \(\widetilde\rho_{ti}=\rho_{ti}=z_{ti}/a_{ti}\).
With multiple copies, the visibility probability can be
strictly smaller than this target.
Part~\ref{item:proposal-capacity-aggregation} verifies the
condition for half-revenue extraction in
Theorem~\ref{thm:general-static}, uniformly over product
capacities.

\subsection{Revenue Extraction from Independent Proposals}
\label{sec:static-general}

Our policy selects a revenue-maximizing assortment from the
visible products. We compare the resulting expected revenue
with the product-level SBLP benchmark by applying
Theorem~\ref{thm:general-static}.

\textbf{Benchmark and Visible-Set Revenue.}
Fix a period and suppress its time index. All revenues below
are conditional on an arrival; the factor \(p_t\) is restored
in the online performance analysis.
Let \(B\) have the SBLP assortment distribution \(\theta\).
In the notation of Section~\ref{sec:general-static},
its purchase probabilities are \(z_i=\E[\phi_i(B)]\), and
\(z_i=a_i\rho_i\) by Lemma~\ref{lem:aggregate-identities}.
Its expected revenue \(W:=\sum_{i\in\cM}r_i z_i\) is the
period's conditional SBLP benchmark.
The policy does not sample \(B\).

The visible proposal set \(P_{\rm vis}\) contains product \(i\)
if and only if at least one of its copies proposes.
For a realization \(P_{\rm vis}=A\), offering an optimal
assortment yields expected revenue \(f(A)\).
Thus, \(\mathbb E[f(P_{\rm vis})]\) is the expected revenue from offering a revenue-maximizing assortment from the visible products, conditional on a customer arrival. The desired bound follows from nonnegativity when \(W=0\), so we assume \(W>0\).

\textbf{Revenue Guarantee.}
By Lemma~\ref{lem:aggregate-identities}, products enter
\(P_{\rm vis}\) independently with probabilities
\(\widetilde\rho_i\ge\rho_i/(1+\rho_i)\).
Applying Theorem~\ref{thm:general-static} with
\(p_i=\widetilde\rho_i\) yields the following corollary.

\begin{corollary}[Multi-capacity visible-set extraction]
\label{cor:visible-static-general}
The visible proposal set satisfies
\[
\E[f(P_{\rm vis})]
\ge
\left(
\min_{i:z_i>0}
\frac{\widetilde\rho_i}
     {\widetilde\rho_i+\rho_i(1-\widetilde\rho_i)}
\right)
\sum_{i\in\cM}r_i z_i
\ge
\frac12\sum_{i\in\cM}r_i z_i.
\]
\end{corollary}

The first bound retains the actual visibility probabilities
and can give a guarantee strictly stronger than \(1/2\).
The uniform \(1/2\) bound holds regardless of the number of
copies or the allocation of proposal mass across them.
Thus, aggregating multiple copy proposals into a single
visible product introduces no additional factor depending
on the product capacities.

\subsection{Target Updates and Flow Realizability}
\label{subsec:sblp-realizability}

We now construct virtual-state updates that preserve the
interval marginals and independence assumed in
Subsection~\ref{sec:proposal-generation}. The update rule
must ensure that the joint probability of \([a,b]\) being
present and \(I_t\) being removed from it equals the SBLP
flow \(y_{tic}([a,b])\). We first give a direct update rule,
then show that it can be implemented through the proposal
rule followed by a target update with probability
\(p_ta_{ti}\).

\begin{lemma}[Realizability of copy-level interval flows]
\label{lem:copy-flow-realizability}
Fix a feasible copy-level SBLP solution. For each copy \((i,c)\),
initialize \(\widehat{\mathcal F}_{1ic}=\cN\). In period \(t\),
if \(I_t\) is contained in a maximal free interval \([a,b]\) of
\(\widehat{\mathcal F}_{tic}\), remove \(I_t\) from the virtual
state with conditional probability
\begin{equation}
\kappa_{tic}([a,b])
:=
\frac{y_{tic}([a,b])}{x_{tic}([a,b])},
\label{eq:direct-cut-kernel}
\end{equation}
where the ratio is defined as zero when \(x_{tic}([a,b])=0\).
In all other cases, leave the state unchanged. Denote the
resulting state by \(\widehat{\mathcal F}_{t+1,i,c}\).
Then, for every copy \((i,c)\), period \(t\in\cT\), and
interval \([a,b]\in\cQ\),
\begin{equation}
\Prb\bigl([a,b]\sim\widehat{\mathcal F}_{tic}\bigr)
=
x_{tic}([a,b]).
\label{eq:realized-x-marginal}
\end{equation}
If the transition draws are independent across copies and
periods, the virtual states are mutually independent across
copies at the beginning of each period.
\end{lemma}


By \eqref{eq:sblp-state-cap} and \(p_ta_{ti}\le1\), the
transition probability is well defined. If the current interval
marginal equals \(x_{tic}([a,b])\), the joint probability that
\([a,b]\) is present and cut is
\(x_{tic}([a,b])\kappa_{tic}([a,b])=y_{tic}([a,b])\).
The persistence and residual-interval events therefore reproduce
\eqref{eq:sblp-balance}. Induction from the initial condition
gives the prescribed marginals. Independent transition draws
across copies preserve their independence.

\begin{remark}\label{remark:target_update}
To implement the transitions in Lemma~\ref{lem:copy-flow-realizability} through proposals,
we require the following target-update law: conditional on
the current virtual states and realized proposals, the update
indicators of all proposed copies are mutually independent,
and each proposed copy of product \(i\) is updated with
probability \(p_ta_{ti}\).
Nonproposing copies remain unchanged.
Conditional on the current virtual states and averaging
over the proposal draws, the probability of removing \(I_t\)
from its containing maximal free interval \([a,b]\) is
\[
\pi_{tic}([a,b])p_ta_{ti}
=
\frac{y_{tic}([a,b])}{x_{tic}([a,b])}
=
\kappa_{tic}([a,b]),
\]
with the stated zero-denominator convention.
Together with independent proposal draws, this update law
reproduces the direct transitions jointly across copies.
\end{remark}

It remains to couple these virtual updates with customer
purchases so that every purchase can be assigned to an
updated copy. Section~\ref{sec:dynamic-general} constructs such a coupling while
preserving the prescribed joint update law.


\section{Purchase--Update Coupling}
\label{sec:dynamic-general}

We now couple customer purchases with the virtual updates specified in Remark~\ref{remark:target_update}. In each nonterminal period, conditional on the current virtual states and realized proposals, the prescribed update law selects each proposed copy of product \(i\) independently with probability \(p_ta_{ti}\). The coupling must recover this joint law after averaging over purchase outcomes, while ensuring that every purchase is assigned to an updated proposed copy of the purchased product. This construction reproduces the SBLP interval flows and, together with the containment of virtual availability within physical availability, ensures feasible assignment.

\subsection{Purchase--Update Coupling Conditions}
\label{subsec:categorical-product}

Fix a nonterminal period and suppress the time index. Condition
on the current virtual states and realized copy proposals.
Let \(m\) denote the number of visible products, indexed by \(i\in[m]\). For each product \(i\), let \(m_i\ge1\) denote the number of its proposed copies, indexed locally by \(\ell\in[m_i]\). Copies retain their
physical identities for assignment, while each product appears
at most once in the assortment. The static assortment oracle
selects \(S\in\operatorname*{arg\,max}_{A\subseteq[m]}R(A)\),
where \(R(A):=\sum_{i\in A}r_i\phi_i(A)\).
All probabilities below are also conditional on \(S\).

The policy draws an independent gate
\(G\sim\operatorname{Bernoulli}(\gamma)\), with
\(\gamma\in[0,1]\), before observing the arrival outcome.
An arriving customer is offered \(S\) if \(G=1\) and the empty
assortment otherwise. Subsection~\ref{subsec:gamma-coefficient}
specifies \(\gamma\); setting \(\gamma=1\) gives the ungated policy.

Let \(J\in\{0,1,\ldots,m\}\) denote the purchase outcome:
\(J=i\) means that product \(i\) is purchased, and \(J=0\)
denotes no purchase due to no arrival, a closed gate, or
selection of the outside option.
The purchase probabilities are
\(q_i:=\Prb(J=i)=\gamma p\phi_i(S)\) for \(i\in S\),
\(q_i=0\) otherwise, and \(q_0:=1-\sum_{i=1}^m q_i\),
where \(p\) is the arrival probability.

For each proposed copy \(\ell\in[m_i]\) of product \(i\),
let \(X_{i\ell}\) indicate whether it receives a virtual update.
The target-update law in Remark~\ref{remark:target_update}
requires mutually independent indicators
\(X_{i\ell}\sim\operatorname{Bernoulli}(\alpha_i)\),
where \(\alpha_i:=pa_i\) and \(a_i:=\phi_i(\{i\})\).
The gate changes the purchase probabilities but leaves this
joint update law unchanged.
Every purchase must be assigned to an updated proposed copy
of the purchased product. Updates to other copies remove the
requested interval only from their virtual free-slot sets,
without reserving physical capacity.

\textbf{Product-Level Coupling Conditions.}
Define
\(U_i:=\one\{\sum_{\ell=1}^{m_i}X_{i\ell}\ge1\}\)
and \(U:=\{i\in[m]:U_i=1\}\).
Under the prescribed copy-update law, the \(U_i\) are
independent Bernoulli indicators with parameters
\(\eta_i:=1-(1-\alpha_i)^{m_i}\).

We seek a coupling that preserves \(\Prb(J=j)=q_j\)
and ensures that every purchased product belongs to the update set:
\begin{equation}
J=i
\quad\Longrightarrow\quad
i\in U,
\qquad i\in[m],
\quad\text{almost surely}.
\label{eq:purchase-forces-update}
\end{equation}
Averaging over the purchase outcome must recover the prescribed
independent product-update law: for every
\(u=(u_1,\ldots,u_m)\in\{0,1\}^m\),
\begin{equation}
\sum_{\substack{j\in\{0,\ldots,m\}\\q_j>0}}
q_j\Pr\bigl((U_1,\ldots,U_m)=u\mid J=j\bigr)
=
\prod_{i=1}^m\eta_i^{u_i}(1-\eta_i)^{1-u_i}.
\label{eq:coupling-objectives}
\end{equation}
The update indicators need not be independent conditional on \(J\). In particular, observing \(J=i\) forces \(U_i=1\).

In the single-copy MNL setting, \citet{huwu2026} generate virtual updates after observing the customer's choice, ensuring that the copy fulfilling the purchase is updated and recovering the independent update law after averaging over outcomes. We extend this approach to multiple copies and general Markov chain choice by characterizing when the purchase and update distributions admit such a coupling.

For any \(A\subseteq[m]\), \eqref{eq:purchase-forces-update} implies \(\{J\in A\}\subseteq\{U\cap A\ne\varnothing\}\), yielding the necessary inequalities below. The priority-rule representation in the next subsection establishes their sufficiency constructively.

\begin{lemma}[Hall-type coupling characterization]
\label{lem:categorical-product}
Let \(J\in\{0,1,\ldots,m\}\) have probabilities
\((q_0,q_1,\ldots,q_m)\), and let \(U\subseteq[m]\) have
mutually independent Bernoulli\((\eta_i)\) indicators.
A coupling of \(J\) and \(U\) satisfying
\eqref{eq:purchase-forces-update} exists if and only if
\begin{equation}
\sum_{i\in A}q_i
\le
1-\prod_{i\in A}(1-\eta_i),
\qquad A\subseteq[m].
\label{eq:hall-condition}
\end{equation}
\end{lemma}

The two sides of \eqref{eq:hall-condition} are the probability
of purchasing a product in \(A\) and the probability that at
least one product in \(A\) receives an update, respectively.
Regularity and \(m_i\ge1\) give
\(q_i\le\alpha_i\le\eta_i\), so the singleton conditions
always hold. Larger subsets require additional conditions
because purchase events are mutually exclusive, whereas
update events can overlap.

For example, suppose \(p=\gamma=1\), one copy of each of two
products proposes, and both products are offered.
The customer starts at either product with equal probability
and transitions directly to the outside option if that product
is unavailable. Then \(q_1=q_2=\eta_1=\eta_2=1/2\).
Both singleton conditions hold, but the condition for
\(A=\{1,2\}\) fails: a purchase occurs with probability one,
whereas at least one update occurs with probability only
\(3/4\). Hence, no coupling can support every purchase while
preserving the independent update law.

The gate reduces purchase probabilities while preserving
the update distribution.
Subsection~\ref{subsec:gamma-coefficient} shows that
\(\gamma=1-1/e\) guarantees all subset inequalities uniformly
over Markov chain choice models and proposal realizations.
Under MNL, they hold with \(\gamma=1\), so no gate is needed.
Assuming \eqref{eq:hall-condition} holds, we next construct
a sampler conditional on the observed purchase outcome
and lift its product-level updates to individual copies.

\subsection{Priority-Rule Representation and Conditional Sampling}
\label{subsec:priority-sampling}

Lemma~\ref{lem:categorical-product} characterizes when the desired coupling exists. To construct it, we first draw \(U\) from its prescribed independent Bernoulli distribution and use a randomized priority rule to generate an auxiliary purchase outcome. We then reverse the sampling order to generate updates conditional on the actual purchase outcome.

Let \(\pi=(i_1,\ldots,i_k)\) be an ordered list of distinct
products in \([m]\), allowing \(k=0\). We examine the listed
products in order and select the first one belonging to
\(U\). If none belongs to \(U\), we return zero. We call
this selection procedure a \emph{priority rule} and define
\[
\mathsf{sel}_{\pi}(U)
:=
\begin{cases}
i_{\min\{\ell\in[k]:i_\ell\in U\}},
&
U\cap\{i_1,\ldots,i_k\}\ne\varnothing,
\\[2mm]
0,
&
\text{otherwise}.
\end{cases}
\]
For every \(i\in[m]\), \(\mathsf{sel}_{\pi}(U)=i\)
implies \(i\in U\), so each priority rule satisfies the
purchase-support condition. The empty rule always returns
zero.

A single priority rule need not produce the prescribed
purchase probabilities. We therefore randomize over rules.
The following proposition shows that, whenever the Hall-type
conditions hold, a mixture of at most \(m+1\) priority rules
reproduces the target purchase distribution exactly.

\begin{proposition}[Priority-rule decomposition]
\label{prop:priority-decomposition}
Suppose \(U_i\sim\operatorname{Bernoulli}(\eta_i)\)
independently for \(i\in[m]\). For an ordered list
\(\pi=(i_1,\ldots,i_k)\), define
\(v^\pi=(v_1^\pi,\ldots,v_m^\pi)\) by
\[
v^\pi_{i_\ell}
:=
\Prb\bigl(\mathsf{sel}_{\pi}(U)=i_\ell\bigr)
=
\eta_{i_\ell}\prod_{s<\ell}(1-\eta_{i_s}),
\qquad
\ell=1,\ldots,k,
\]
and set \(v_i^\pi:=0\) for products omitted from \(\pi\).
The no-selection probability is
\[
v_0^\pi
:=
\Prb\bigl(\mathsf{sel}_{\pi}(U)=0\bigr)
=
\prod_{\ell=1}^k(1-\eta_{i_\ell})
=
1-\sum_{i=1}^m v_i^\pi,
\]
where an empty product equals one.

If \(q=(q_1,\ldots,q_m)\in\R_+^m\) satisfies
\eqref{eq:hall-condition}, then there exist priority rules
\(\pi^1,\ldots,\pi^L\) and weights
\(\lambda_1,\ldots,\lambda_L\ge0\), with \(L\le m+1\)
and \(\sum_{h=1}^L\lambda_h=1\), such that
\begin{equation}
q
=
\sum_{h=1}^L\lambda_hv^{\pi^h}.
\label{eq:priority-decomposition}
\end{equation}
For rational \(q\) and \(\eta\), such a decomposition can be
computed in polynomial time using a greedy priority-rule oracle.
The fixed irrational gate is handled by the rational reduction
in Appendix~\ref{app:exact-gate-computation}.
\end{proposition}

This representation follows from the classical greedy
description of polymatroids \citep{edmonds2003submodular},
applied to the coverage function
\(g(A)=1-\prod_{i\in A}(1-\eta_i)\).
Appendix~\ref{app:priority-decomposition} gives the proof
and the decomposition algorithm.

Given this decomposition, construct the auxiliary joint
distribution as follows. Draw the indicators \(U_i\)
independently with probabilities \(\eta_i\). Independently
of \(U\), draw \(H\in[L]\) with
\(\Prb(H=h)=\lambda_h\), and set
\[
J^{\rm f}:=\mathsf{sel}_{\pi^H}(U).
\]
Conditional on \(H=h\), the outcome probabilities are
\((v_0^{\pi^h},v_1^{\pi^h},\ldots,v_m^{\pi^h})\).
The decomposition and normalization therefore give
\[
\Prb(J^{\rm f}=j)
=
\sum_{h=1}^L\lambda_hv_j^{\pi^h}
=
q_j,
\qquad j=0,1,\ldots,m.
\]
The update indicators retain their prescribed independent
Bernoulli distribution, and every priority rule ensures
that \(J^{\rm f}=i\) implies \(i\in U\).
Thus, \((J^{\rm f},U)\) has both required marginal
distributions and satisfies the purchase-support condition.

\textbf{Conditional Sampling After the Customer Outcome.}
The auxiliary outcome \(J^{\rm f}\) defines the joint law;
the online policy observes the actual purchase outcome \(J\)
and then samples from \(\mathcal L(H,U\mid J^{\rm f}=j)\)
when \(J=j\).
For an observed outcome \(J=j\) with \(q_j>0\), Bayes' rule gives
\[
\Prb(H=h\mid J^{\rm f}=j)
=
\frac{\lambda_hv_j^{\pi^h}}{q_j},
\qquad h\in[L].
\]
After drawing \(H\), the policy conditions the update indicators
on the selected rule returning \(j\).
If \(j=i>0\), it sets \(U_i=1\) and the indicators of products
preceding \(i\) in the rule to zero.
If \(j=0\), it sets all indicators listed in the rule to zero.
In either case, the remaining indicators are sampled independently
with their original Bernoulli parameters.
Thus, products omitted from the rule may receive updates even
when no purchase occurs.

Because \(J\) and \(J^{\rm f}\) have the same distribution,
this conditional sampling reproduces the auxiliary joint law
after averaging over \(J\).
Algorithm~\ref{alg:conditional-coupling} in
Appendix~\ref{app:conditional-coupling} gives the detailed procedure.

\begin{proposition}[Correctness of the posterior sampler]
\label{prop:conditional-coupling}
Suppose \eqref{eq:hall-condition} holds and \(J\) has
probabilities \((q_0,q_1,\ldots,q_m)\). Given a decomposition
satisfying \eqref{eq:priority-decomposition},
Algorithm~\ref{alg:conditional-coupling} has the following
properties:
\begin{enumerate}[label=(\roman*)]
\item
For every \(i\in[m]\), \(J=i\) implies \(i\in U\)
almost surely.

\item
After averaging over \(J\), the indicators
\((U_i)_{i\in[m]}\) are mutually independent
Bernoulli\((\eta_i)\) variables. Equivalently, their
distribution satisfies \eqref{eq:coupling-objectives}.
\end{enumerate}

Given a priority decomposition of size \(L\) as in Proposition~\ref{prop:priority-decomposition}, the posterior weights for all positive-probability outcomes can be precomputed in \(O((m+1)L)\) arithmetic operations. Each execution uses \(O(m+L)\) arithmetic operations, one categorical draw, and at most \(m\) Bernoulli draws. Since \(L\le m+1\), preprocessing requires \(O(m^2)\) arithmetic operations and each execution requires \(O(m)\) arithmetic operations and random draws.
\end{proposition}


\subsection{Copy-Level Updates and Physical Assignment}
\label{subsec:copy-lifting}
The product-level sampler determines whether each product
has at least one proposed copy selected for a virtual update.
We now select the individual copies consistently with these
outcomes. Given \(U\), generate each product's update vector
using its own indicator \(U_i\) and independent random draws
across products:
\begin{itemize}
\item
If \(U_i=0\), set \(X_{i\ell}=0\) for every \(\ell\in[m_i]\).

\item
If \(U_i=1\), sample \((X_{i1},\ldots,X_{im_i})\) from
the joint law of \(m_i\) independent
Bernoulli\((\alpha_i)\) variables, conditional on at least
one success. Thus, for every nonzero vector
\(x\in\{0,1\}^{m_i}\),
\begin{equation}
\Prb\bigl((X_{i1},\ldots,X_{im_i})=x\mid U_i=1\bigr)
=
\frac{
\alpha_i^{\lVert x\rVert_1}
(1-\alpha_i)^{m_i-\lVert x\rVert_1}
}{\eta_i},
\label{eq:copy-lifting-law}
\end{equation}
where \(\eta_i=1-(1-\alpha_i)^{m_i}\).
\end{itemize}
If \(\eta_i=0\), then \(U_i=0\) almost surely, so the
conditional law in \eqref{eq:copy-lifting-law} is never used.

After averaging over \(J\), the product indicators \(U_i\) have their prescribed independent distribution. For each product, mixing the zero-update and nonempty-update branches recovers the joint law of independent Bernoulli\((\alpha_i)\) copy updates. Because the lifting uses independent random draws across products, the resulting updates are independent across all proposed copies.
Appendix~\ref{app:copy-update-law} gives the proof and a
direct sampler requiring \(O(m_i)\) operations per product.

If \(J=i\), Proposition~\ref{prop:conditional-coupling} ensures \(U_i=1\), so at least one proposed copy of product \(i\) is selected for an update. Assign the purchase to one such copy and reserve the requested interval on it. Regardless of the purchase outcome, remove the requested interval from the virtual free-slot set of every proposed copy selected for an update, across all products. Only the assigned copy receives a physical reservation; all other updates are virtual discards. When \(J=0\), every update is a virtual discard.

Thus, whenever \eqref{eq:hall-condition} holds, the construction
preserves the independent copy-update law and assigns every
purchase to an updated proposed copy.
Subsection~\ref{subsec:gamma-coefficient} selects a gate
that guarantees these conditions.

\section{Online Policy and Performance Guarantees}
\label{sec:performance}

We now combine the proposal mechanism and purchase--update
coupling into an online policy. We first choose the gate
probability to ensure that the coupling conditions hold
for every proposal realization. We then describe the policy
and establish its feasibility and approximation guarantees.

\subsection{Gate Selection}
\label{subsec:gamma-coefficient}

The coupling constructed in Section~\ref{sec:dynamic-general} requires
the Hall-type inequalities in Lemma~\ref{lem:categorical-product}.
We now choose the gate probability so that these inequalities
hold for every proposal realization, while preserving the
prescribed independent copy-update law.

Fix a nonterminal period and suppress the time index. Treat the
current virtual states, realized copy proposals, and preselected
assortment \(S\) as fixed. For each visible product \(i\), let
\(m_i\ge1\) denote the number of its proposed copies. Each
proposed copy must receive a virtual update with probability
\(pa_i\), where \(a_i:=\phi_i(\{i\})\), independently of the
other proposed copies. Thus, product \(i\) receives at least
one update with probability
\(\eta_i:=1-(1-pa_i)^{m_i}\).

Recall that the policy draws an independent gate
\(G\sim\operatorname{Bernoulli}(\gamma)\) before observing
the arrival outcome and offers \(S\) if the gate is open
and the customer arrives. A closed gate results in an empty
offered assortment. The resulting purchase probabilities are
\(q_i=\gamma p\phi_i(S)\) for \(i\in S\) and zero for
visible products outside \(S\). The coupling must recover the
prescribed independent copy-update law after averaging over
the gate, arrival, and choice outcomes summarized by \(J\).

By Lemma~\ref{lem:categorical-product}, the gate is feasible
if and only if
\begin{equation}
\gamma p\sum_{i\in A}\phi_i(S)
\le
1-\prod_{i\in A}(1-\eta_i),
\qquad A\subseteq S.
\label{eq:gated-hall-condition}
\end{equation}
It suffices to consider subsets of \(S\): products outside
\(S\) have zero purchase probability, and adding them to
a subset can only increase its probability of receiving
at least one update. Because the gate is independent of
arrival and choice, the expected revenue conditional on
an arrival is \(\gamma R(S)\).


For every visible product, \(m_i\ge1\) implies
\(\eta_i=1-(1-pa_i)^{m_i}\ge pa_i\).
Thus, for any \(A\subseteq S\),
\(
1-\prod_{i\in A}(1-\eta_i)
\ge
1-\prod_{i\in A}(1-pa_i).
\)
The right-hand side is the update probability when each
product in \(A\) has exactly one proposed copy. A gate
that satisfies all Hall-type inequalities in this case
therefore remains feasible with any larger number of
proposed copies.

\begin{definition}[Product-update compatibility]
\label{def:gamma-choice}
For a fixed choice model \(\phi\), define
\begin{equation}
\Gamma(\phi)
:=
\min
\left\{
1,
\inf_{\substack{
p\in(0,1],\ S\subseteq\cM,\ A\subseteq S,
\sum_{i\in A}\phi_i(S)>0
}}
\frac{
1-\prod_{i\in A}(1-pa_i)
}{
p\sum_{i\in A}\phi_i(S)
}
\right\}.
\label{eq:gamma-choice}
\end{equation}
\end{definition}

Each ratio in \eqref{eq:gamma-choice} compares the probability
of at least one update in \(A\), with one proposed copy per
product, to the ungated purchase probability of \(A\).
Thus, \(\Gamma(\phi)\) is the largest gate acceptance
probability that satisfies all these one-copy inequalities,
uniformly over \(p\), \(S\), and \(A\).
Every \(\gamma\in[0,\Gamma(\phi)]\) consequently satisfies
\eqref{eq:gated-hall-condition} for every proposal realization.


The following result gives a fixed gate that is valid for
every Markov chain choice model.

\begin{proposition}[Universal Markov coupling gate]
\label{prop:universal-gamma}
For every Markov chain choice model,
\begin{equation}
\Gamma(\phi)
\ge
1-\frac1e.
\label{eq:universal-gamma}
\end{equation}
Consequently, the fixed gate \(\gamma=1-1/e\) satisfies
the Hall-type conditions in
Lemma~\ref{lem:categorical-product} for every proposal
realization and every assortment selected from the
visible products.
\end{proposition}

\proof{Proof.}
Fix \(p\in(0,1]\), \(S\subseteq\cM\),
and \(A\subseteq S\), and set \(x:=p\sum_{i\in A}a_i\).
Regularity gives \(\phi_i(S)\le a_i\), while the total
purchase probability conditional on an arrival is at
most one. Hence,
\(
p\sum_{i\in A}\phi_i(S)
\le
\min\{p,x\}
\le
\min\{1,x\}.
\)
Using \(1-y\le e^{-y}\), we obtain
\[
\begin{aligned}
1-\prod_{i\in A}(1-pa_i)
\ge 1-e^{-x}
\ge \left(1-\frac1e\right)\min\{1,x\} \ge \left(1-\frac1e\right)
p\sum_{i\in A}\phi_i(S).
\end{aligned}
\]
The second inequality follows from the concavity of
\(1-e^{-x}\) on \([0,1]\) and its monotonicity for
\(x\ge1\). Every ratio in \eqref{eq:gamma-choice} is
therefore at least \(1-1/e\), proving
\eqref{eq:universal-gamma}. Additional proposed copies
only increase the probability of at least one update,
so the same gate remains feasible.
\hfill\(\blacksquare\)
\endproof

This proof uses only regularity and the fact that total
purchase probability is at most one. The gate bound
therefore applies to any regular discrete-choice model.

\textit{Sharpness of the uniform gate.}
The constant \(1-1/e\) in Proposition~\ref{prop:universal-gamma}
is the largest scalar gate that works uniformly over all
Markov chain choice models.

\begin{proposition}[Sharpness of the universal Markov gate]
\label{prop:gamma-sharp}
No scalar gate \(\gamma>1-1/e\) satisfies the required Hall-type
inequalities uniformly over all Markov chain choice models.
Hence, the constant \(1-1/e\) in
Proposition~\ref{prop:universal-gamma} is best possible among
uniform scalar gates.
\end{proposition}

This result concerns gates that must be valid throughout the
Markov family. Gates tailored to the realized instance or
different proposal and update mechanisms may permit stronger
guarantees.

\textit{Lossless coupling under MNL.}
Under the positive-outside normalization in
\eqref{eq:mnl-choice}, MNL is a special case of Markov
chain choice. Its choice probabilities satisfy a stronger
compatibility property: all Hall-type inequalities hold
without gating. Consequently, the coupling can be
implemented with \(\gamma=1\).

\begin{proposition}[Lossless product-update coupling under MNL]
\label{prop:mnl-gamma-one}
Under the positive-outside normalization in \eqref{eq:mnl-choice},
\(\Gamma(\phi)=1\). Thus, the customer outcome can be coupled with
the prescribed independent product-update law without attenuating
the purchase probabilities.
\end{proposition}

These results specify the gate probabilities used below:
\(\gamma=1-1/e\) under general Markov chain choice and
\(\gamma=1\) under MNL. We now combine these gates with
the proposal and coupling constructions.

\subsection{Online Policy}
\label{subsec:general-policy}

The policy uses an optimal SBLP solution to generate candidate
products in each period. It first draws independent proposals
from virtually feasible copies and makes a product visible
whenever one of its copies proposes. It then selects a revenue-maximizing assortment from the visible products. In each nonterminal period, the policy applies the selected gate and, after observing the purchase outcome, samples virtual updates and assigns any purchase to an updated proposed copy of the purchased product. Together, the updates associated with purchases and virtual discards reproduce the SBLP interval flows used to generate future proposals.

\textbf{Physical and virtual availability.}
For each copy \((i,c)\), the policy maintains an actual
free-slot set \(\mathcal F_{tic}\) and a virtual free-slot
set \(\overline{\mathcal F}_{tic}\), both initialized to \(\cN\).
The actual set changes only when a purchase is assigned to
that copy; the virtual set determines its proposals.
Removing an interval only from the virtual set is called a
\emph{virtual discard}.
The policy maintains
\(\overline{\mathcal F}_{tic}\subseteq\mathcal F_{tic}\),
so every proposed copy can physically fulfill the request.
In nonterminal periods, the updates of
\(\overline{\mathcal F}_{tic}\) implement the transitions
defined for \(\widehat{\mathcal F}_{tic}\) in
Lemma~\ref{lem:copy-flow-realizability}.

For example, in a nonterminal period, suppose two copies of a
product propose and both receive virtual updates. If a purchase is assigned to the first
copy, \(I_t\) is removed from both virtual sets but only from
the first copy's actual set. The second copy remains physically
available for \(I_t\); its update is a virtual discard.

\textbf{Decisions and updates.}
Let \(\mathcal O_t(A)\) return a revenue-maximizing assortment
contained in \(A\). We use \(\gamma=1-1/e\) under general
Markov chain choice and \(\gamma=1\) under MNL, as justified
by Propositions~\ref{prop:universal-gamma}
and~\ref{prop:mnl-gamma-one}, respectively.
Remark~\ref{rem:mnl-zero-outside} covers the MNL boundary case.

For each nonterminal period, we use the purchase--update
coupling developed above as a subroutine.
Given the proposals and selected assortment, it forms the
purchase and update probabilities from
Subsection~\ref{subsec:categorical-product} and computes the
priority decomposition in
Proposition~\ref{prop:priority-decomposition}.
It then applies Algorithm~\ref{alg:conditional-coupling}
conditional on the observed purchase outcome and lifts the
product updates to individual copies using
Subsection~\ref{subsec:copy-lifting}.
Appendix~\ref{app:exact-gate-computation} handles the fixed
irrational gate.

Let \(W_{ti}\) denote the labels of proposed copies of
product \(i\) selected for a virtual update.
The coupling ensures \(W_{ti}\ne\varnothing\) whenever
\(J_t=i\), allowing the purchase to be assigned to an updated
copy. Updates are sampled even when \(J_t=0\); in that case,
all updates are virtual discards.

The terminal period requires no gate or coupling because
no further proposal distributions need to be preserved.
Any purchase is assigned directly to a proposed copy, and
the requested interval is removed from both of that copy's
states. Let \(P_t\) denote the set of products with at least one proposed copy in period \(t\). When \(P_t=\varnothing\), the offered assortment and all update sets are empty.

In the terminal period, one could instead optimize over all
physically available products \(\cA_T\) and assign any purchase
to a feasible copy. Since \(P_T\subseteq\cA_T\), we have
\(f_T(\cA_T)\ge f_T(P_T)\), so this modification weakly
increases revenue and preserves the approximation guarantees.
Algorithm~\ref{alg:general-policy} retains proposal-based
assortment selection in every period for a uniform policy
definition.

\begin{algorithm}[t]
\caption{Multi-Capacity Markov Proposal-Discarding Policy}
\label{alg:general-policy}
\begin{algorithmic}[1]
\footnotesize
\REQUIRE An optimal copy-level SBLP solution \((x,y,z)\),
assortment oracles \((\mathcal O_t)\), and gate probability
\(\gamma\).

\STATE Initialize
\(\mathcal F_{1ic}=\overline{\mathcal F}_{1ic}=\cN\)
for every copy \((i,c)\).

\FOR{\(t=1,\ldots,T\)}
    \STATE Copy all current actual and virtual states
    to period \(t+1\).

    \STATE Generate the proposed-copy set \(P_t^c\)
    independently across virtually feasible copies using
    \eqref{eq:copy-proposal}.

    \STATE Form
    \(P_t:=\{i:\exists c,\ (i,c)\in P_t^c\}\)
    and select \(\widehat S_t\leftarrow\mathcal O_t(P_t)\).

    \STATE If \(t<T\), independently draw
    \(G_t\sim\operatorname{Bernoulli}(\gamma)\);
    otherwise set \(G_t\leftarrow1\).

    \STATE Set \(S_t\leftarrow\widehat S_t\) if \(G_t=1\),
    and \(S_t\leftarrow\varnothing\) otherwise.

    \STATE Observe the arrival. If a customer arrives,
    offer \(S_t\) and observe \(J_t\in S_t\cup\{0\}\);
    otherwise set \(J_t\leftarrow0\).

    \IF{\(t<T\)}
        \STATE Sample \((W_{ti})_{i\in P_t}\) conditional on
        \(J_t\) using Proposition~\ref{prop:priority-decomposition},
        Algorithm~\ref{alg:conditional-coupling}, and the lifting in
        Subsection~\ref{subsec:copy-lifting}, including when \(J_t=0\).
    \ELSE
        \STATE Set \(W_{ti}\leftarrow\varnothing\)
        for every \(i\in P_t\).
        \STATE If \(J_t>0\), choose any proposed copy
        \((J_t,c)\in P_t^c\) and set
        \(W_{t,J_t}\leftarrow\{c\}\).
    \ENDIF

    \IF{\(J_t>0\)}
        \STATE Set \(i\leftarrow J_t\), assign the purchase
        to any \(c^*\in W_{ti}\), and remove \(I_t\)
        from \(\mathcal F_{t+1,i,c^*}\).
    \ENDIF

    \STATE For every \(i\in P_t\) and \(c\in W_{ti}\),
    remove \(I_t\) from
    \(\overline{\mathcal F}_{t+1,i,c}\).
\ENDFOR
\end{algorithmic}
\end{algorithm}

\textbf{MNL assortment oracles.}
The required assortment oracles are efficient under MNL.
For nonnegative product weights, an optimal assortment
can be chosen as a prefix of the products ordered by
nonincreasing weight
\citep{talluri2004general,gallego2004managing}.
Sorting the products and evaluating the prefixes using
cumulative sums requires \(O(|A|\log|A|)\) arithmetic
and comparison operations for a candidate set \(A\). This implements both the weighted
assortment oracle used by the SBLP and the
revenue-maximizing oracle used by the online policy.

\subsection{State Invariants and Physical Feasibility}
\label{subsubsec:state-invariants}

We first verify that the coupling produces the prescribed
independent copy updates. We then use this update law to
establish the virtual-state invariants and physical feasibility.

\begin{lemma}[Exact copy-level update law]
\label{lem:copy-update-law-general}
Fix a nonterminal period \(t<T\), and condition on the current
virtual states, realized copy proposals, and preselected
assortment \(\widehat S_t\). After averaging over \(J_t\),
the update indicators of all proposed copies are mutually
independent. Each proposed copy of product \(i\) receives
a virtual update with probability \(\alpha_{ti}=p_ta_{ti}\).
\end{lemma}

Lemma~\ref{lem:copy-update-law-general} follows from
Proposition~\ref{prop:conditional-coupling} and the copy-level
lifting in Subsection~\ref{subsec:copy-lifting}, which together
recover the prescribed joint update law.

\begin{theorem}[Multi-capacity implementation invariants]
\label{thm:implementation-general}
Algorithm~\ref{alg:general-policy} satisfies the following properties:
\begin{enumerate}[label=(\roman*),ref=(\roman*)]
\item
\label{item:invariant-containment}
For every copy \((i,c)\) and \(t=1,\ldots,T+1\),
\(\overline{\mathcal F}_{tic}\subseteq\mathcal F_{tic}\)
along every sample path.

\item
Every purchase is assigned to a physically feasible copy of
the purchased product, and the requested interval is reserved
on exactly one physical copy.

\item
\label{item:invariant-marginals}
At the beginning of every period \(t\in\cT\),
\(
\Prb\bigl([a,b]\sim\overline{\mathcal F}_{tic}\bigr)
=
x_{tic}([a,b])
\)
for every copy \((i,c)\) and interval \([a,b]\in\cQ\).

\item
\label{item:invariant-independence}
At the beginning of every period \(t\in\cT\), the virtual
free-slot sets are mutually independent across copies.
\end{enumerate}
Consequently, actual occupancy never exceeds physical capacity.
\end{theorem}

The argument proceeds by induction over periods. Initially, all
actual and virtual copies are fully available, so properties
\ref{item:invariant-containment}, \ref{item:invariant-marginals}, and \ref{item:invariant-independence} hold by the initialization and
\eqref{eq:sblp-boundary}.

\textit{Virtual-state evolution.}
Consider a nonterminal period \(t<T\) and a virtually feasible
copy whose maximal free interval \([a,b]\) contains \(I_t\).
Conditional on the current virtual states, its probability
of proposing and receiving an update is
\[
\pi_{tic}([a,b])p_ta_{ti}
=
\frac{y_{tic}([a,b])}{x_{tic}([a,b])},
\]
where the ratio is defined as zero when its denominator is zero.
This is the transition probability in
Lemma~\ref{lem:copy-flow-realizability}.
The proposal draws are independent across copies, and
Lemma~\ref{lem:copy-update-law-general} gives independent
updates conditional on the realized proposals. Since each
proposed copy's update probability \(p_ta_{ti}\) does not
depend on which other copies propose, the joint transition
law factors across copies and each copy's transition depends
only on its own virtual state. The virtual states therefore
retain the prescribed interval marginals and independence
through the beginning of period \(T\).

\textit{Physical feasibility.}
Suppose virtual containment holds at the beginning of period
\(t\). A copy proposes only if its virtual state contains
\(I_t\), so every proposed copy is physically feasible.
In a nonterminal period, the coupling assigns every
purchase to an updated proposed copy. The policy removes
\(I_t\) from both states of the assigned copy and only from
the virtual states of other updated copies, preserving
virtual containment. If no purchase occurs, actual states
remain unchanged and virtual states can only shrink.
In the terminal period, a purchase is assigned directly to
a proposed copy and \(I_t\) is removed from both its states.
Thus, every purchase reserves the requested interval on
exactly one physically feasible copy, and containment holds
through period boundary \(T+1\).

The distributional statements in \ref{item:invariant-marginals} and \ref{item:invariant-independence} average over
past outcomes; they do not condition on a realized purchase
history. 


\subsection{Approximation Guarantees}
\label{subsec:theoretical-results}

We prove both parts of Theorem~\ref{thm:main-guarantees} using a common argument. Fix the optimal SBLP solution \((x,y,z)\) used by Algorithm~\ref{alg:general-policy}. By Theorem~\ref{thm:implementation-general}, the virtual states at the beginning of each period have the prescribed interval marginals and are independent across copies. Corollary~\ref{cor:visible-static-general} therefore gives \(\mathbb E[f_t(P_t)]\ge \frac12\sum_{i\in\mathcal M}r_{ti}z_{ti}\) for every period \(t\), where the expectation averages over the virtual states and proposals.
The current arrival is independent of the candidate set. In each nonterminal period, the gate retains a fraction \(\gamma\) of the expected assortment revenue, and the coupling preserves the resulting purchase probabilities. The terminal period uses no gate. Consequently,
\[
\begin{aligned}
V^\pi(I)
=
\gamma\sum_{t=1}^{T-1}p_t\E[f_t(P_t)]
+p_T\E[f_T(P_T)]
\ge
\frac{\gamma}{2}
\sum_{t=1}^{T}p_t\sum_{i\in\cM}r_{ti}z_{ti}
=
\frac{\gamma}{2}\SBLP^C(I).
\end{aligned}
\]
The inequality follows from Corollary~\ref{cor:visible-static-general}, \(\gamma\le1\), and nonnegative revenues. The final equality follows from optimality of the SBLP solution. Proposition~\ref{prop:universal-gamma} permits \(\gamma=1-1/e\) under general Markov chain choice, while Proposition~\ref{prop:mnl-gamma-one} permits \(\gamma=1\) under MNL. Substituting these values and applying the SBLP upper bound on \(V^{\rm on}(I)\) proves \eqref{eq:markov-main} and \eqref{eq:mnl-main}.


The following result shows that the \(1/2\) guarantee is tight
for the proposed MNL policy relative to the optimal online
benchmark.

\begin{proposition}[Tightness for the proposed MNL policy]
\label{prop:mnl-policy-tightness}
There exists a sequence of one-period, unit-capacity MNL instances
\((I_m)_{m\ge1}\) such that
\begin{equation}
\frac{V^{\pi_{\rm MNL}}(I_m)}{V^{\rm on}(I_m)}
\longrightarrow
\frac12.
\label{eq:mnl-policy-tightness}
\end{equation}
Hence, the \(1/2\) guarantee in Theorem~\ref{thm:main-guarantees}(ii) is tight for the proposed MNL policy.
\end{proposition}

This tightness result concerns Algorithm~\ref{alg:general-policy}, which uses independent proposals even in the terminal period. In a one-period instance, the optimal online policy offers a revenue-maximizing assortment from all available products, while the proposed policy first restricts the candidate set through independent proposals. The limiting loss therefore arises from this proposal step. The terminal-period modification described above is optimal on these one-period instances and is not covered by this tightness result.

\textit{Running time.}
The aggregate SBLP is solved once before the booking horizon,
as established in Proposition~\ref{prop:sblp-polytime-general}. It contains \(O(TMN^2)\)
interval variables, and forming the symmetric copy-level
interval targets requires \(O(TN^2C_{\mathrm{tot}})\)
arithmetic operations.

In each period, the policy solves an assortment problem over
the realized visible set \(P_t\). This problem is polynomial-time
solvable under Markov chain choice and requires
\(O(m\log(2+m))\) arithmetic and comparison operations under
MNL, where \(m:=|P_t|\).
In each nonterminal period, the policy also computes a priority
decomposition for the induced purchase and update probabilities,
using Proposition~\ref{prop:priority-decomposition} and Appendix~\ref{app:exact-gate-computation}.
Given this decomposition, for \(m\ge1\), computing all posterior
weights takes \(O(m^2)\) arithmetic operations, and one
execution of the conditional sampler takes \(O(m)\) operations.
The copy-level sampling then takes
\(O(\sum_{i\in P_t}m_{ti})\) operations.
These sampling counts treat random draws as primitives and
exclude the polynomial-time computation of the decomposition.
The assortment oracle and, in nonterminal periods, the
priority decomposition are the main online optimization tasks.

The remaining state-maintenance operations are polynomial
in the model dimensions. Under the computational convention
stated at the beginning of the electronic companion, the bound
\(C_{\mathrm{tot}}\le MT\) and the exact-sampling construction
in Appendix~\ref{app:exact-gate-computation} imply that
the complete policy has expected running time polynomial
in \(T,N,M\) and the model input size.

\begin{remark}
When \(C_i=1\) for every product, product visibility
coincides with the unique copy's proposal, whose marginal
probability is \(\rho_{ti}=z_{ti}/a_{ti}\).
In each nonterminal period, conditional on the current
virtual states and realized proposals, the prescribed
update probability of a proposing copy is \(p_ta_{ti}\).
This specialization retains the resource-level
proposal-discarding framework of \citet{huwu2026}, without
any copy aggregation. Since the MNL coupling requires no
gate, the improvement from \(0.271\) to \(1/2\) follows
from the stronger static revenue-extraction analysis.
\end{remark}

\begin{remark}[Zero outside-option attractiveness]
\label{rem:mnl-zero-outside}
If the outside option has zero attractiveness in every period,
the same ungated MNL policy satisfies the stronger guarantee
\[
V^{\pi_{\rm MNL}}(I)
\ge
\left(1-\frac1e\right)\SBLP^C(I)
\ge
\left(1-\frac1e\right)V^{\rm on}(I).
\]
In this case, \(a_{ti}=1\), and a highest-revenue singleton
is optimal within any nonempty candidate set. Since the
independent copy-proposal probabilities satisfy
\(\sum_c\rho_{tic}=z_{ti}\) and \(\sum_i z_{ti}\le1\),
a direct bound on the probability of seeing a product above
each revenue threshold gives
\(\E[f_t(P_t)]\ge(1-1/e)\sum_i r_{ti}z_{ti}\).
The Hall-type inequalities remain valid with \(\gamma=1\),
so the coupling introduces no further loss.
Our policy attains the same \(1-1/e\) guarantee as
\citet{huwu2026} when the outside option has zero
attractiveness in every period. Appendix~\ref{app:mnl-zero-outside}
establishes this guarantee for arbitrary copy multiplicities
and shows that the factor is tight for the policy using the
symmetric lift.
\end{remark}


\section{Conclusion}
\label{sec:conclusion}

We study online assortment optimization under Markov chain choice
when each booking reserves consecutive service slots on a single
physical copy and each product may have multiple interchangeable
copies. We formulate a copy-expanded SBLP and show that aggregating
copies of the same product preserves its optimal value. The
resulting target flows guide a proposal-discarding policy that
includes each product at most once in each offered assortment
while retaining copy identities for assignment. By keeping each
virtual free-slot set within its actual counterpart, the policy
ensures that every purchase is assigned to a copy available
throughout the requested interval.

Our analysis separates static revenue extraction from dynamic
implementation. A Markov-path argument shows that optimizing
over products made available by independent copy proposals
recovers at least half of the SBLP revenue benchmark in
expectation. A coverage-polymatroid coupling then coordinates
customer purchases with the prescribed independent virtual
updates. A uniform gate with acceptance probability \(1-1/e\)
ensures the feasibility of this coupling under general Markov
choice, whereas MNL requires no gating. Together, these results
yield approximation guarantees of \((1-1/e)/2\) and \(1/2\),
respectively, relative to the optimal nonanticipative online
value, for arbitrary product capacities. Under MNL with zero
outside-option attractiveness in every period, the guarantee
improves to \(1-1/e\).

The \(1/2\) extraction factor and the \(1-1/e\) uniform
scalar-gate factor are separately sharp, and neither loss
is an artifact of copy aggregation. These sharpness results
do not establish tightness of the combined Markov guarantee,
since the two worst cases need not coincide. Under MNL, the
\(1/2\) guarantee is tight for the proposed policy even with
one copy per product, but this policy-specific result does
not rule out stronger guarantees for other policies. Further improvements may come from adapting the gate
to the realized proposals or jointly redesigning proposal
generation and virtual updates.


{\small
\SingleSpacedXI
\bibliographystyle{informs2014}
\setlength{\parskip}{0pt}
\setlength{\bibsep}{4pt plus 0.3ex}
\bibliography{ref}
}

\RUNTITLE{Multi-Capacity Online Assortment Optimization under Markov Choice}
\clearpage
\normalsize
\pagestyle{ECheadings}
\makeatletter
\ECHowEquations
\ECHowSections
\setcounter{figure}{0}
\renewcommand\thefigure{EC.\@arabic\c@figure}
\setcounter{table}{0}
\renewcommand\thetable{EC.\@arabic\c@table}
\setcounter{page}{1}
\renewcommand\thepage{ec\arabic{page}}
\makeatother
\onecolumn
\begin{center}
\large
\textbf{{\sc Electronic Companion}\\
\textquotedblleft Multi-Capacity Online Assortment Optimization with Consecutive Stays under Markov Chain Choice\textquotedblright}
\end{center}
\vspace{5mm}
\begin{APPENDICES}
\small


\section{Proofs of Results in Section~\ref{sec:sblp}}
\label{sec:proof-sblp}

\subsection{Proof of Proposition~\ref{prop:sblp-upper}}

\proof{Proof.}
We construct a feasible copy-expanded SBLP solution from any
feasible nonanticipative policy and show that its objective
equals the policy's expected revenue.

\textit{Step 1: Define the assortment and purchase variables.}
Fix a feasible nonanticipative policy \(\pi\).
Let \(\mathscr H_t\) be the information sigma-field immediately
before the period-\(t\) arrival, including the observed history
and the policy randomization used to select the assortment.
The current free-slot sets and the preselected assortment
\(S_t\) are therefore \(\mathscr H_t\)-measurable.
Let \(A_t\) denote the arrival indicator, and let \(J_t\)
denote the purchased product, with \(J_t=0\) when no arrival
or no purchase occurs. By the arrival and choice assumptions,
\(\Prb_\pi(A_t=1\mid\mathscr H_t)=p_t\) and
\(\Prb_\pi(A_t=1,J_t=i\mid\mathscr H_t) = p_t\phi_{ti}(S_t)\).
After observing a purchase, the policy may use additional
randomization to select a feasible copy.

Define
\(\theta_{tS}:=\Prb_\pi(S_t=S)\) and
\(z_{ti}:=\Prb_\pi(J_t=i\mid A_t=1)\).
Because the assortment is selected before the independent
arrival, its distribution conditional on an arrival remains
\(\theta_t\). Consequently,
\(\sum_{S\subseteq\cM}\theta_{tS}=1\) and \(z_{ti} =
\sum_{S\subseteq\cM}\theta_{tS}\phi_{ti}(S)\).
Using the convention \(\phi_{ti}(S)=0\) for \(i\notin S\),
these identities establish
\eqref{eq:sblp-theta}--\eqref{eq:sblp-choice}.

\textit{Step 2: Define the copy-level variables and verify their bounds.}
For each copy \((i,c)\) and interval \(K\in\cQ\), let
\(H_{tic}(K):=\one\{K\sim\mathcal F_{tic}\}\).
Let \(D_{tic}(K)\) indicate that the period-\(t\) customer
arrives, purchases product \(i\), and is assigned to copy
\((i,c)\), with \(K\) as the maximal free interval containing
the requested interval. Set
\(x_{tic}(K):=\E_\pi[H_{tic}(K)]\) and
\(y_{tic}(K):=\E_\pi[D_{tic}(K)]\).

Every purchase is assigned to exactly one copy, and its
consecutive service interval lies in a unique maximal free
interval of that copy. Thus, pathwise,
\(\sum_{c=1}^{C_i}\sum_{K\in\cQ}D_{tic}(K) = \one\{A_t=1,J_t=i\}\).
Taking expectations gives
\(\sum_c\sum_K y_{tic}(K)=p_tz_{ti}\), which is
\eqref{eq:sblp-sales-link}.

The event \(D_{tic}(K)=1\) requires \(I_t\subseteq K\).
Moreover,
\[
\begin{aligned}
y_{tic}(K)
\le
\E_\pi\!\left[
H_{tic}(K)\one\{A_t=1,J_t=i\}
\right]=
p_t\E_\pi\!\left[
H_{tic}(K)\phi_{ti}(S_t)
\right]
\le
p_ta_{ti}x_{tic}(K).
\end{aligned}
\]
The equality follows by conditioning on \(\mathscr H_t\),
since \(H_{tic}(K)\) is \(\mathscr H_t\)-measurable.
The final inequality follows from
\eqref{eq:singleton-upper}.
Together with the containment requirement, this establishes
\eqref{eq:sblp-state-cap}.
The argument allows the assortment to depend on the copy's
availability and allows assignment randomization after
the purchase.

\textit{Step 3: Verify interval balances and initial conditions.}
Fix an interval \(K=[a,b]\).
For every \(t<T\), the interval indicators satisfy the
pathwise identity
\begin{equation}
\begin{aligned}
H_{t+1,i,c}([a,b])
={}&
H_{tic}([a,b])-D_{tic}([a,b])\\
&+
\one\{b+1=\ell_t\}
\sum_{\substack{
b'\ge r_t\\
{[a,b']\in\cQ}
}}
D_{tic}([a,b'])+
\one\{a-1=r_t\}
\sum_{\substack{
a'\le\ell_t\\
{[a',b]\in\cQ}
}}
D_{tic}([a',b]).
\end{aligned}
\label{eq:app-interval-identity}
\end{equation}

To verify this identity, first suppose that copy \((i,c)\)
is not used in period \(t\). All of its maximal free intervals
then persist. Otherwise, the reservation removes exactly
one parent interval. A parent starting at \(a<\ell_t\)
creates the nonempty left residual \([a,\ell_t-1]\),
and a parent ending at \(b>r_t\) creates the nonempty right
residual \([r_t+1,b]\).

These residuals are maximal free intervals. Their inner
endpoints are adjacent to newly reserved slots, while their
outer endpoints inherit maximality from the parent.
Removing service slots cannot merge free intervals or alter
any other maximal free interval.

For a fixed interval \([a,b]\), persistence and creation
as a residual are mutually exclusive. Although a reservation
may create two different residuals, it cannot create the
same interval twice. Thus,
\eqref{eq:app-interval-identity} counts each possible source
of the interval exactly once.
Taking expectations yields \eqref{eq:sblp-balance},
with the time index shifted by one.
Full initial availability gives
\(x_{1ic}(K)=\one\{K=[1,N]\}\), establishing
\eqref{eq:sblp-boundary}.

\textit{Step 4: Compare objective values.}
All induced variables are nonnegative, and their objective
value satisfies
\(\sum_{t,i,c,K}r_{ti}y_{tic}(K) = \sum_{t,i}r_{ti}\Prb_\pi(A_t=1,J_t=i) =
V^\pi(I)\).
Hence, \(V^\pi(I)\le\SBLP^C(I)\) for every feasible
nonanticipative policy \(\pi\).
Taking the supremum over these policies proves
\eqref{eq:sblp-upper}.\hfill\(\blacksquare\)
\endproof

\subsection{Proof of Proposition~\ref{prop:sblp-symmetry}}

\proof{Proof.}
Write \(K\in\cQ\) for a service interval.
We show that aggregation and symmetric lifting preserve
both feasibility and objective value.

\textit{Aggregation.}
Let \((x,y,\theta,z)\) be a feasible copy-expanded solution.
Define
\(X_{ti}(K):=\sum_{c=1}^{C_i}x_{tic}(K)\) and
\(Y_{ti}(K):=\sum_{c=1}^{C_i}y_{tic}(K)\),
leaving \(\theta\) and \(z\) unchanged.

The choice constraints remain satisfied because they
involve only \(\theta\) and \(z\).
The sales-linking constraints give
\(\sum_{K\in\cQ}Y_{ti}(K) = \sum_{c=1}^{C_i}\sum_{K\in\cQ}y_{tic}(K) =
p_tz_{ti}\).
The interval-capacity coefficients are identical across
copies of the same product. Summing the corresponding
constraints therefore yields
\(0\le Y_{ti}(K) \le p_ta_{ti}X_{ti}(K)\one\{I_t\subseteq K\}\).
Similarly, summing the copy-level balance equations gives
\eqref{eq:aggregate-balance}, and summing the initial
conditions gives
\(X_{1i}(K)=C_i\one\{K=[1,N]\}\).
Thus, the aggregate solution is feasible.
Its objective is unchanged because
\(\sum_{t,i,c,K}r_{ti}y_{tic}(K)
=\sum_{t,i,K}r_{ti}Y_{ti}(K)\).

\textit{Symmetric lifting.}
Conversely, let \((X,Y,\theta,z)\) be a feasible aggregate
solution. For every \(c\in[C_i]\), define
\(x_{tic}(K):=X_{ti}(K)/C_i\) and
\(y_{tic}(K):=Y_{ti}(K)/C_i\),
again leaving \(\theta\) and \(z\) unchanged.
These variables are well defined and nonnegative because
\(C_i\ge1\).

Dividing each aggregate interval-capacity constraint and
balance equation by \(C_i\) gives the corresponding
copy-level constraint.
The aggregate initial condition likewise gives
\(x_{1ic}(K)=\one\{K=[1,N]\}\).
Summing the lifted variables over copies recovers \(X\)
and \(Y\); in particular,
\(\sum_{c=1}^{C_i}\sum_{K\in\cQ}y_{tic}(K) = \sum_{K\in\cQ}Y_{ti}(K) =
p_tz_{ti}\).
Thus, the sales-linking constraints and objective value
are also preserved.

The two formulations consequently have the same attainable
objective values and hence the same optimal value.

\textit{Number of interval variables.}
There are \(|\cQ|=N(N+1)/2\) possible intervals and two
aggregate interval variables, \(X_{ti}(K)\) and \(Y_{ti}(K)\),
for each \((t,i,K)\).
Their total number is
\(2TM|\cQ|=TMN(N+1)=O(TMN^2)\),
independently of the capacities.
This count concerns only the interval state and sales-flow
variables. The assortment-mixture variables are handled
in the proof of
Proposition~\ref{prop:sblp-polytime-general}.\hfill\(\blacksquare\)
\endproof

\subsection{Proof of Proposition~\ref{prop:sblp-polytime-general}}

\proof{Proof.}

Let \(L\) denote the size of the numerical model input.
Since \(C_i\le T\), at most \(MT\) copies are represented.
We first derive the aggregate LP dual and construct its
separation oracle. We then establish polynomial solvability,
recover an optimal primal solution with sparse assortment
mixtures, and construct the symmetric copy-level lift.

\textit{Step 1: Formulate the aggregate LP and its dual.}
Let \(u\) collect the aggregate interval variables
\((X_{ti}(K),Y_{ti}(K))\).
Write the interval balance equations and initial conditions
as \(Eu=b\), and the interval-capacity constraints as
\(Fu\le0\).
Specifically, the row of \(Fu\) indexed by \((t,i,K)\) is
\(Y_{ti}(K)
-p_ta_{ti}\one\{I_t\subseteq K\}X_{ti}(K)\).
The matrices \(E,F\) and vector \(b\) are explicitly specified
by the aggregate constraints and have polynomial size.

Define the linear map \(H\) and objective coefficient
vector \(g\) by
\((Hu)_{ti}:=\sum_{K\in\cQ}Y_{ti}(K)\) and \(g^{\mathsf
T}u:=\sum_{t,i,K}r_{ti}Y_{ti}(K)\).
Since \(p_t>0\), the sales-linking constraints allow us
to eliminate \(z\) by setting
\(z_{ti}=(Hu)_{ti}/p_t\).
The resulting primal formulation is
\begin{equation}
\begin{aligned}
\max_{u\ge0,\;\theta\ge0}\quad
&g^{\mathsf T}u\\
\textnormal{s.t.}\quad
&Eu=b,\qquad Fu\le0,\\
&(Hu)_{ti}
=
p_t\sum_{S\subseteq\cM}
\theta_{tS}\phi_{ti}(S),
&&t\in\cT,\ i\in\cM,\\
&\sum_{S\subseteq\cM}\theta_{tS}=1,
&&t\in\cT.
\end{aligned}
\label{eq:app-aggregate-primal}
\end{equation}
Here and below, \(\phi_{ti}(S)=0\) for \(i\notin S\).
This LP has polynomially many rows and interval-variable
columns, but exponentially many assortment columns
indexed by \(\theta_{tS}\).
We access the assortment columns implicitly through
the oracle constructed below.

Associate unrestricted multipliers \(\xi\) with the
interval equalities, nonnegative multipliers \(\mu\)
with the interval-capacity inequalities, unrestricted
multipliers \(\beta_{ti}\) with the sales-linking equalities,
and unrestricted multipliers \(\alpha_t\) with the
assortment-normalization equalities.
Choose the signs of the equality multipliers so that
the Lagrangian is
\[
\begin{aligned}
\mathcal L
={}
b^{\mathsf T}\xi+\sum_t\alpha_t+
\bigl(
g-E^{\mathsf T}\xi-F^{\mathsf T}\mu-H^{\mathsf T}\beta
\bigr)^{\mathsf T}u+
\sum_t\sum_{S\subseteq\cM}
\theta_{tS}
\left(
p_t\sum_{i\in\cM}\beta_{ti}\phi_{ti}(S)-\alpha_t
\right).
\end{aligned}
\]
Its supremum over \(u,\theta\ge0\) is finite precisely
when every coefficient of these variables is nonpositive.
The dual is therefore
\begin{equation}
\begin{aligned}
\min_{\xi,\mu,\beta,\alpha}\quad
&b^{\mathsf T}\xi+\sum_t\alpha_t\\
\textnormal{s.t.}\quad
&E^{\mathsf T}\xi+F^{\mathsf T}\mu+H^{\mathsf T}\beta
\ge g,\\
&\alpha_t
\ge
p_t\sum_{i\in\cM}\beta_{ti}\phi_{ti}(S),
&&t\in\cT,\ S\subseteq\cM,\\
&\mu\ge0,\qquad
\xi,\beta,\alpha\text{ unrestricted}.
\end{aligned}
\label{eq:app-aggregate-dual}
\end{equation}

Each assortment variable \(\theta_{tS}\) thus corresponds
to one assortment-specific dual inequality.
The dual has polynomially many variables and explicit
constraints; only the assortment-specific family is
exponential.

\textit{Step 2: Construct a separation oracle for the dual.}
Consider a candidate dual solution.
Its explicit constraints can be checked directly.
For a fixed period \(t\), define
\(d_i:=p_t\beta_{ti}\).
All assortment-specific inequalities for this period
hold if and only if
\(\max_{S\subseteq\cM} \sum_{i\in S}d_i\phi_{ti}(S) \le\alpha_t\).
An optimizing assortment therefore either identifies a
violated inequality or certifies the entire family
for period \(t\).

The weights \(d_i\) may be negative.
Let \(A^+:=\{i\in\cM:d_i>0\}\), and for any assortment
\(S\), put \(S^+:=S\cap A^+\).
Nonnegativity of purchase probabilities and regularity give
\(\sum_{i\in S}d_i\phi_{ti}(S) \le \sum_{i\in S^+}d_i\phi_{ti}(S) \le
\sum_{i\in S^+}d_i\phi_{ti}(S^+)\).
Hence, an optimal assortment can be chosen within \(A^+\).
If \(A^+=\varnothing\), the empty assortment is optimal
with value zero. Otherwise, separation reduces to
\eqref{eq:weighted-static} with nonnegative rewards.

We now give an exact implementation of this weighted
assortment oracle. Suppress the period index, let
\(A\subseteq\cM\) be the eligible products, and let
\(w_i\ge0\) be their rewards.
Consider the stopping problem in which the chain may stop
at a product \(i\in A\) for reward \(w_i\), must continue
at products outside \(A\), and receives zero upon reaching
the outside option.
Its value is characterized by
\begin{equation}
\begin{aligned}
\min_{v\ge0}\quad
&\sum_{j\in\cM}v_j\\
\textnormal{s.t.}\quad
&v_j\ge\sum_{k\in\cM}q_{jk}v_k,
&&j\in\cM,\\
&v_j\ge w_j,
&&j\in A.
\end{aligned}
\label{eq:app-stopping-lp}
\end{equation}
The inequalities require \(v_j\) to dominate the value
of continuation and, at eligible products, the reward
from stopping.

To verify this characterization, start from \(v^{(0)}=0\)
and apply the Bellman operator:
\[
v_j^{(n+1)}
=
\begin{cases}
\max\{w_j,\sum_{k\in\cM}q_{jk}v_k^{(n)}\},
&j\in A,\\
\sum_{k\in\cM}q_{jk}v_k^{(n)},
&j\notin A.
\end{cases}
\]
The iterates are componentwise nondecreasing and bounded
above by
\(W:=\max(\{0\}\cup\{w_i:i\in A\})\).
Every feasible vector of \eqref{eq:app-stopping-lp}
dominates every iterate, by induction.
Their limit \(v^*\) satisfies the Bellman equations and
is the componentwise least feasible vector.
It is therefore the unique minimizer of
\eqref{eq:app-stopping-lp}.

Define \(S^*:=\{j\in A:v_j^*=w_j\}\).
At every product outside \(S^*\), the Bellman equations give
\(v_j^*=\sum_{k\in\cM}q_{jk}v_k^*\).
Repeated substitution expresses \(v_j^*\) as the expected
reward from reaching \(S^*\cup\{0\}\) within finitely many
transitions, plus a remainder bounded by \(W\) times the
probability of not yet reaching this set.
Transience of the product-state submatrix makes the
remainder tend to zero.
Thus, stopping at the first visit to \(S^*\cup\{0\}\)
attains \(v_j^*\) from every initial product state.

Conversely, any feasible vector of
\eqref{eq:app-stopping-lp} bounds the expected revenue
from every assortment contained in \(A\).
It dominates the reward at offered products and the
continuation value elsewhere; a finite-step expansion
followed by transience gives the corresponding
infinite-horizon bound.
Starting from the initial distribution \(\lambda\),
we therefore obtain
\(\max_{S\subseteq A}\sum_{i\in S}w_i\phi_i(S) =
\sum_{j\in\cM}\lambda_jv_j^* = \sum_{i\in S^*}w_i\phi_i(S^*)\).
Solving \eqref{eq:app-stopping-lp} and returning \(S^*\)
therefore implements the required oracle, consistent
with \citet{blanchet2016markov}.
Products outside \(A\) remain transition states of the
chain; they are only ineligible for offering.

The stopping LP has \(M\) variables and at most \(2M\)
constraints apart from nonnegativity.
For rational inputs, it can be solved exactly in time
polynomial in the model input and query size.
The value iteration above is used only to prove the
characterization; the oracle is computed by solving the LP.
The choice probabilities of a returned assortment can
also be computed exactly in polynomial time, as shown
in Step~3.
Applying the oracle to each period therefore provides
a polynomial-time strong separation oracle for
\eqref{eq:app-aggregate-dual}.

\textit{Step 3: Establish polynomial solvability.}
We first verify uniform bounds on the numerical descriptions
of the LP coefficients.
For an assortment \(S\), let \(D=\cM\setminus S\).
Transience implies that \(I-Q_{t,DD}\) is nonsingular.
Writing \(\lambda_{t,D}\) as a column vector, we have
\(\phi_{ti}(S) = \lambda_{ti} + \lambda_{t,D}^{\mathsf T}
(I-Q_{t,DD})^{-1}Q_{t,Di}\) for \(i\in S\).
The second term is absent when \(D=\varnothing\).
Rational linear algebra and determinant bounds show
that these probabilities have encoding length polynomial
in \(M\) and \(L\), uniformly over all assortments.
The singleton probabilities \(a_{ti}\) can likewise be
computed exactly in polynomial time.
Thus, every coefficient in the primal and dual has
polynomial encoding length, without enumerating the
assortments.

The primal is feasible: set all sales variables to zero,
offer the empty assortment with probability one, and keep
\(X_{ti}(K)=C_i\one\{K=[1,N]\}\) throughout the horizon.
It is also bounded.
Multiplying each interval balance equation by interval
length and summing over intervals gives
\[
\sum_{K\in\cQ}|K|X_{t+1,i}(K)
=
\sum_{K\in\cQ}|K|X_{ti}(K)
-
|I_t|\sum_{K\in\cQ}Y_{ti}(K),
\qquad t<T.
\]
Indeed, a feasible cut removes exactly \(|I_t|\) service
slots from its parent interval.
The initial weighted sum is \(NC_i\), so nonnegativity
implies \(0\le X_{ti}(K)\le NC_i\).
The choice-mixture and sales-linking constraints also give
\(0\le \sum_{K\in\cQ}Y_{ti}(K) = p_t\sum_{S\subseteq\cM}
\theta_{tS}\phi_{ti}(S) \le p_t\le1\).
Together with \(0\le\theta_{tS}\le1\), these bounds show
that the primal has a finite attained optimum.
Strong duality therefore applies.

The dual has polynomial dimension, uniformly bounded
coefficient encoding lengths, and the strong separation
oracle from Step~2.
Standard rational-polyhedral bounds also ensure the
existence of an optimal dual solution of polynomial
encoding length.
The rational ellipsoid method consequently computes
the optimal value and an optimal dual solution in time
polynomial in \(T,N,M,L\)
\citep{grotschel1981ellipsoid}.
This exact rational framework also accommodates feasible
regions contained in a proper affine subspace.

\textit{Step 4: Recover an optimal primal solution and sparse mixtures.}
The online policy requires the interval flows, so computing
the optimal value alone is insufficient.
We use the constructive primal-recovery procedure associated
with the rational ellipsoid method to identify polynomially
many assortment columns that support an optimal primal
solution.

To explain why these columns suffice, let \(R^*\) be the
optimal value from Step~3.
The primal has polynomially many rows, and all column
coefficients have uniformly bounded encoding lengths.
A basic feasible solution uses at most as many basic
variables as there are rows.
Determinant bounds therefore give a computable polynomial
bound \(B\) such that every finite optimal value has reduced
denominator at most \(2^B\), both for the full primal and
for any restriction to a subset of its assortment columns.
Two distinct such optimal values consequently differ by
at least \(2^{-2B}\).

Set \(\varepsilon:=2^{-2B-1}\).
The dual constraints together with
\(b^{\mathsf T}\xi+\sum_t\alpha_t \le R^*-\varepsilon\)
form an infeasible rational system.
The certificate-producing rational ellipsoid method
identifies polynomially many assortment inequalities
that, together with the explicit dual constraints and
this objective bound, already certify infeasibility.
This is the constructive primal-recovery step of the
optimization--separation framework
\citep{grotschel1981ellipsoid}.

Any auxiliary bounding box used in this procedure can
be chosen uniformly for all relevant subsystems.
Indeed, their dimensions and coefficient encoding lengths
are polynomially bounded, so every feasible subsystem
has a feasible point of polynomial encoding length.
A sufficiently large common box therefore preserves
feasibility of every such subsystem.
Consequently, the infeasibility certificate remains
valid after the auxiliary box is removed.

For each period \(t\), let \(\mathcal S_t\) consist of the
assortments appearing in the certificate, together with
the empty assortment.
Restrict \eqref{eq:app-aggregate-primal} to variables
\(\theta_{tS}\) with \(S\in\mathcal S_t\).
This restricted primal is feasible because it includes
the empty assortment.
Its dual retains precisely the corresponding
assortment-specific inequalities, together with all
explicit dual constraints.

Let \(R_{\mathcal S}\) be the restricted primal optimum.
The infeasibility certificate and strong duality imply
\(R_{\mathcal S}>R^*-\varepsilon\).
Because restricting the available columns cannot increase
the primal optimum, \(R_{\mathcal S}\le R^*\).
Thus,
\(R^*-\varepsilon<R_{\mathcal S}\le R^*\).
The rational-value separation bound forces
\(R_{\mathcal S}=R^*\).
Solving this polynomial-size restricted primal therefore
recovers optimal interval flows and an optimal assortment
mixture without enumerating all assortments.

Recover \(z_{ti}=(Hu)_{ti}/p_t\).
For each period, we can reduce the mixture support further
while holding \(u\) and \(z\) fixed.
Over the recovered family \(\mathcal S_t\), consider
\(\sum_{S\in\mathcal S_t}\theta_{tS}=1\), \(\sum_{S\in\mathcal
S_t}\theta_{tS}\phi_{ti}(S)=z_{ti}\) for \(i\in\cM\), and
\(\theta_{tS}\ge0\).
This system is feasible and bounded.
Its equality matrix has rank at most \(M+1\), so a basic
feasible solution has at most \(M+1\) positive weights.
Such a solution can be obtained by polynomial-time linear
programming and basis extraction.
Replacing the recovered mixture by this basic solution
preserves the purchase probabilities, interval flows,
and objective value.
We output only its positive weights, with all unlisted
assortment weights understood to be zero.

\textit{Step 5: Construct the copy-level lift.}
By Proposition~\ref{prop:sblp-symmetry}, the optimal aggregate
solution yields an optimal copy-expanded solution through
\(x_{tic}(K)=X_{ti}(K)/C_i\) and
\(y_{tic}(K)=Y_{ti}(K)/C_i\) for \(c\in[C_i]\),
with \(\theta\) and \(z\) unchanged.
Writing \(C_{\rm tot}:=\sum_i C_i\), the lift contains
\(2T|\cQ|C_{\rm tot}=O(TN^2C_{\rm tot})\)
interval entries, each of polynomial encoding length.

Solving the aggregate relaxation takes time polynomial
in \(T,N,M,L\).
Explicitly constructing the lift adds a dependence on
\(C_{\rm tot}\).
Since \(C_{\rm tot}\le MT\), the full procedure runs
in time polynomial in \(T,N,M,L\), as claimed. \hfill\(\blacksquare\)
\endproof

\section{Proofs of Results in Section~\ref{sec:general-static}}
\label{sec:proof-general-static}

\subsection{Proof of Proposition~\ref{prop:markov-extraction-sharpness}}\label{app:proof-markov-extraction-sharpness}

\proof{Proof.}
\emph{Step 1: Construct the choice model.}
Fix an integer \(m\ge2\). Consider an MNL model with
\(m^2\) products, each with unit revenue and attractiveness
\(u_i=1/(m-1)\), and normalize the outside-option
attractiveness to one. An assortment \(S\) has expected
revenue
\[
R(S)=\frac{|S|}{m-1+|S|}.
\]
This expression is strictly increasing in \(|S|\), so
offering every available product is optimal. Consequently,
\(f(A)=|A|/(m-1+|A|)\) for every candidate set
\(A\subseteq\cM\), including \(A=\emptyset\).

\emph{Step 2: Specify the benchmark and candidate probabilities.}
Let the benchmark offer the full assortment
\(B=\cM\) with probability one, and write
\(d_m:=m^2+m-1\). For every product,
\[
a_i=\frac1m,
\qquad
z_i=\frac1{d_m},
\qquad
W=\sum_{i\in\cM}r_i z_i=\frac{m^2}{d_m}.
\]
Let \(P\) include each product independently with probability
\[
p_i=\rho_i=\frac{z_i}{a_i}
=\frac{m}{d_m}
=:\rho_m.
\]
Thus, the model has equal positive revenues, the benchmark
is deterministic, and the candidate probabilities satisfy \(p_i=\rho_i=z_i/a_i\).

\emph{Step 3: Bound the extraction ratio.}
Let \(K_m:=|P|\). Then
\(K_m\sim\operatorname{Binomial}(m^2,\rho_m)\), with
\(\E[K_m]=m^3/d_m\), and
\[
f(P)=\frac{K_m}{m-1+K_m}.
\]
The function \(x\mapsto x/(m-1+x)\) is concave on
\([0,\infty)\). Jensen's inequality therefore gives
\[
\frac{\E[f(P)]}{W}
\le
\frac{1}{W}
\frac{\E[K_m]}{m-1+\E[K_m]}
=
\frac{m d_m}{(m-1)d_m+m^3}
\longrightarrow \frac12.
\]
Hence, for every \(\varepsilon>0\), choosing \(m\)
sufficiently large yields
\(\E[f(P)]/W\le1/2+\varepsilon\), proving the
proposition. Since Theorem~\ref{thm:general-static} also gives
\(\E[f(P)]/W\ge1/2\), the extraction ratio converges
to \(1/2\) as \(m\to\infty\).

\hfill\(\blacksquare\)
\endproof

\section{Proofs of Results in Section~\ref{sec:proposal-revenue}}
\label{sec:proof-proposal-revenue}

We verify the proposal identities, apply the general extraction
bound to the visible products, and establish the target interval flows.

\subsection{Proof of Lemma~\ref{lem:aggregate-identities}}
\proof{Proof.}
For part~(i), \eqref{eq:sblp-sales-link} and
\eqref{eq:rho-copy} give
\[
p_tz_{ti}
=\sum_{c=1}^{C_i}\sum_{[a,b]\in\cQ}y_{tic}([a,b])
=p_ta_{ti}\rho_{ti}.
\]
Since \(p_t>0\), we obtain \(z_{ti}=a_{ti}\rho_{ti}\).
The choice-mixture constraints and the singleton bound give
\(z_{ti}\le a_{ti}\), so \(\rho_{ti}\le1\).
Nonnegativity then gives
\(0\le\rho_{tic}\le\rho_{ti}\le1\).

For part~(ii), each copy's proposal depends only on its own
virtual state and independent random draw.
The state independence assumption therefore implies independence
of the copy-proposal indicators, whose marginals are given by
\eqref{eq:rho-copy}.
Product \(i\) is invisible precisely when none of its copies
proposes, yielding \eqref{eq:visible-rho}.
Visibility events for distinct products depend on disjoint
groups of independent copy proposals and are therefore
independent.

For part~(iii), \eqref{eq:visible-rho} and
\(1-u\le e^{-u}\) imply
\[
1-\widetilde\rho_{ti}
=\prod_{c=1}^{C_i}(1-\rho_{tic})
\le\exp\!\left(-\sum_{c=1}^{C_i}\rho_{tic}\right)
=e^{-\rho_{ti}}.
\]
The remaining inequalities follow from \(e^u\ge1+u\)
and \(\rho_{ti}\le1\), respectively.
\hfill\(\blacksquare\)
\endproof

\subsection{Proof of Corollary~\ref{cor:visible-static-general}}

\proof{Proof.}
Fix a period, suppress its time index, and write
\(W:=\sum_{i\in\cM}r_i z_i\). If \(W=0\), the revenue
lower bound follows from nonnegativity. Suppose \(W>0\).

Lemma~\ref{lem:aggregate-identities} gives
\(\rho_i=z_i/a_i=\sum_{c=1}^{C_i}\rho_{ic}\).
By Lemma~\ref{lem:aggregate-identities}, the product-inclusion
indicators of \(P_{\rm vis}\) are mutually independent,
with
\(\Prb(i\in P_{\rm vis}) = \widetilde\rho_i =
1-\prod_{c=1}^{C_i}(1-\rho_{ic})\).
Applying Theorem~\ref{thm:general-static} with
\(P=P_{\rm vis}\) and \(p_i=\widetilde\rho_i\) yields
\[
\E[f(P_{\rm vis})]
\ge
\left(
\min_{i:z_i>0}
\frac{\widetilde\rho_i}
     {\widetilde\rho_i+\rho_i(1-\widetilde\rho_i)}
\right)W,
\]
which proves the first inequality.

For the second inequality,
Lemma~\ref{lem:aggregate-identities} gives
\(\widetilde\rho_i\ge\rho_i/(1+\rho_i)\).
The final assertion of Theorem~\ref{thm:general-static}
therefore implies
\(\E[f(P_{\rm vis})] \ge \frac12 W = \frac12\sum_{i\in\cM}r_i z_i\).
The argument allows different proposal probabilities
across copies of the same product.\hfill\(\blacksquare\)
\endproof

\subsection{Proof of Lemma~\ref{lem:copy-flow-realizability}}

\proof{Proof.}
\textit{Step 1: Verify that the transition is well defined.}
Fix a feasible copy-level SBLP solution and a copy
\((i,c)\). By \eqref{eq:sblp-state-cap} and
\(p_ta_{ti}\le1\), we have
\(0\le y_{tic}(K)\le x_{tic}(K)\) for every \(K\in\cQ\).
Thus, the ratio
\(\kappa_{tic}(K)=y_{tic}(K)/x_{tic}(K)\), defined as
zero when \(x_{tic}(K)=0\), is a valid probability.
The same constraint gives \(y_{tic}(K)=0\) whenever
\(I_t\nsubseteq K\).

A consecutive request can be contained in at most one
maximal free interval of a given state. The rule therefore
specifies a unique possible cut in each state and leaves
the state unchanged when the request is infeasible.

\textit{Step 2: Establish the prescribed interval marginals.}
We prove \eqref{eq:realized-x-marginal} by induction.
At \(t=1\), the virtual state is \(\cN\), whose only
maximal free interval is \([1,N]\). Its interval
marginals therefore agree with
\eqref{eq:sblp-boundary}.

Suppose the required marginals hold at the beginning
of period \(t\). Let \(\widehat D_{tic}(K)\) indicate
that \(K\) is a maximal free interval of
\(\widehat{\mathcal F}_{tic}\) and the transition
removes \(I_t\) from \(K\). The induction hypothesis
and the transition rule give
\(\Prb\bigl(\widehat D_{tic}(K)=1\bigr) = x_{tic}(K)\kappa_{tic}(K) =
y_{tic}(K)\).
If \(x_{tic}(K)=0\), both the interval-event probability
and its sales flow are zero. If \(I_t\nsubseteq K\),
both the cut probability and the sales flow are zero.
Thus, the identity holds for every \(K\in\cQ\).

For \(t<T\), apply the pathwise interval identity
\eqref{eq:app-interval-identity} to the virtual state
and its cut indicators. Taking expectations yields
\[
\begin{aligned}
\Prb\bigl([a,b]\sim\widehat{\mathcal F}_{t+1,i,c}\bigr)
={}&
x_{tic}([a,b])-y_{tic}([a,b])\\
&+
\one\{b+1=\ell_t\}
\sum_{\substack{b'\ge r_t\\{[a,b']\in\cQ}}}
y_{tic}([a,b'])+
\one\{a-1=r_t\}
\sum_{\substack{a'\le\ell_t\\{[a',b]\in\cQ}}}
y_{tic}([a',b])\\
={}&
x_{t+1,i,c}([a,b]),
\end{aligned}
\]
where the final equality follows from
\eqref{eq:sblp-balance}. This completes the induction
through the beginning of period \(T\).

If the auxiliary transition is also applied in period
\(T\), the same argument gives terminal cut probability
\(y_{Tic}(K)\). No marginal identity involving
\(x_{T+1,i,c}\) is required.

\textit{Step 3: Establish independence and connect the
transition to proposals.}
Implement the direct transitions using random draws
independent across copies and periods. Each copy's
virtual-state trajectory then depends only on its
deterministic initial state and its own draws.
The trajectories, and hence the virtual states at
each fixed period, are mutually independent across copies.

For a copy whose virtual state contains \(I_t\) in
maximal free interval \(K\), proposing with probability
\(\pi_{tic}(K)\) and then receiving a target update
with probability \(p_ta_{ti}\) gives cut probability
\(\pi_{tic}(K)p_ta_{ti} = \frac{y_{tic}(K)}{x_{tic}(K)} = \kappa_{tic}(K)\),
with the stated convention when \(x_{tic}(K)=0\).
Under the target copy-update law in
Remark~\ref{remark:target_update}, these updates are
independent across proposed copies conditional on
the current virtual states and realized proposals.
Each proposed copy's update probability depends only
on its product, not on the other proposals. Averaging
over the independent proposal draws therefore gives
the same joint transition law as the independent
direct-cut construction.

Theorem~\ref{thm:implementation-general} verifies that
the purchase--update coupling implements this target
law in every nonterminal period.\hfill\(\blacksquare\)
\endproof

\section{Proofs of Results in Section~\ref{sec:dynamic-general}}
\label{sec:proof-dynamic-general}

These proofs establish the period-level purchase--update coupling
and its exact computational implementation.

\subsection{Proof of Lemma~\ref{lem:categorical-product}}

\proof{Proof.}
Write
\(\mu(B):=\Prb(U=B) = \prod_{i\in B}\eta_i \prod_{i\notin B}(1-\eta_i)\) for
\(B\subseteq[m]\).
If \(m=0\), the deterministic pair
\((J,U)=(0,\varnothing)\) satisfies the claim.
We therefore assume \(m\ge1\).

\textit{Necessity.}
The support condition \(J=i\Rightarrow i\in U\) implies
\(\{J\in A\}\subseteq\{U\cap A\ne\varnothing\}\).
Consequently, for every \(A\subseteq[m]\),
\(\sum_{i\in A}q_i \le \Prb(U\cap A\ne\varnothing) = 1-\prod_{i\in
A}(1-\eta_i)\).

\textit{Sufficiency.}
Construct a bipartite flow network with a source and a sink.
There is a left node for each \(j\in\{0,1,\ldots,m\}\),
connected from the source by an edge of capacity \(q_j\).
There is a right node for each \(B\subseteq[m]\),
connected to the sink by an edge of capacity \(\mu(B)\).
Connect a left node \(i>0\) to \(B\) whenever \(i\in B\),
and connect node \(0\) to every right node.
Give these intermediate edges unlimited capacity.
The total capacity leaving the source and the total
capacity entering the sink both equal one.

For a set \(A\) of left nodes, let \(N(A)\) denote its
right neighbors. A finite cut that places exactly \(A\)
of the left nodes on the source side must also place
\(N(A)\) on that side. The minimum capacity of such a
cut is
\(1-\sum_{j\in A}q_j+\sum_{B\in N(A)}\mu(B)\).
The max-flow/min-cut theorem therefore gives a flow of
value one if and only if
\(\sum_{j\in A}q_j \le \sum_{B\in N(A)}\mu(B)\)
for every set \(A\) of left nodes.

If \(0\in A\), then \(N(A)\) contains every right node,
so its total mass is one and the inequality holds
automatically. If \(0\notin A\), then
\(N(A)=\{B:B\cap A\ne\varnothing\}\), and the inequality
is precisely \eqref{eq:hall-condition}.
Thus, the assumed Hall inequalities guarantee a
value-one flow.

Such a flow saturates all source and sink capacities.
Assigning the flow on edge \((j,B)\) as the probability
of \((J,U)=(j,B)\) gives the prescribed marginals
\((q_j)\) and \((\mu(B))\). The intermediate edges ensure
that \(J=i>0\) implies \(i\in U\).

This network establishes existence. The priority-rule
representation below provides a construction that
does not enumerate all update sets. \hfill\(\blacksquare\)
\endproof

\subsection{Proof of Proposition~\ref{prop:priority-decomposition}}\label{app:priority-decomposition}

\proof{Proof.}
We establish the convex-hull representation through
greedy optimization and then recover a decomposition
with small support.

Let \(\mathfrak P\) denote the collection of ordered
subsets of \([m]\), including the empty list. Define
\(g(A):=1-\prod_{i\in A}(1-\eta_i)\) and
\[
P(g):=
\left\{
x\in\R_+^m:
\sum_{i\in A}x_i\le g(A)
\text{ for every }A\subseteq[m]
\right\}.
\]
The function \(g\) is normalized and monotone. For
\(i\notin A\), its marginal increment is
\(g(A\cup\{i\})-g(A) = \eta_i\prod_{j\in A}(1-\eta_j)\),
which is nonincreasing as \(A\) grows. Thus, \(g\) is
submodular and \(P(g)\) is a polymatroid.
The priority vectors are the classical greedy vectors
of this polymatroid \citep{edmonds2003submodular}.
We give the specialization and its computational
implementation below.

\textit{Step 1: Verify feasibility of the priority vectors.}
Fix \(\pi\in\mathfrak P\), generate the independent
Bernoulli update indicators, and apply the priority rule.
If the selected product belongs to \(A\), then
\(U\cap A\ne\varnothing\). Hence,
\(\sum_{i\in A}v_i^\pi\le g(A)\) for every \(A\),
so \(v^\pi\in P(g)\).

For \(\pi=(i_1,\ldots,i_k)\), let
\(A_\ell:=\{i_1,\ldots,i_\ell\}\), with
\(A_0=\varnothing\). The priority probabilities satisfy
\(v_{i_\ell}^\pi = g(A_\ell)-g(A_{\ell-1}) =
\eta_{i_\ell}\prod_{s<\ell}(1-\eta_{i_s})\).
In particular,
\(\sum_{s=1}^{\ell}v_{i_s}^\pi=g(A_\ell)\).

\textit{Step 2: Optimize a linear objective by a priority rule.}
Fix \(c\in\R^m\). List the indices with \(c_i>0\) as
\(\pi=(i_1,\ldots,i_k)\), in nonincreasing order of
weight, and use the convention \(c_{i_{k+1}}:=0\).
For every \(x\in P(g)\), nonnegativity and telescoping
give
\[
\begin{aligned}
c^{\mathsf T}x
\le
\sum_{\ell=1}^k c_{i_\ell}x_{i_\ell}=
\sum_{\ell=1}^k
(c_{i_\ell}-c_{i_{\ell+1}})
\sum_{s=1}^{\ell}x_{i_s}
\le
\sum_{\ell=1}^k
(c_{i_\ell}-c_{i_{\ell+1}})g(A_\ell)=
c^{\mathsf T}v^\pi.
\end{aligned}
\]
The last equality follows from the prefix identities
in Step 1. If no weight is positive, the empty rule
returns the optimal vector \(v^\varnothing=0\).
Therefore,
\begin{equation}
c^{\mathsf T}v^\pi
=
\max_{x\in P(g)}c^{\mathsf T}x
=
\max_{\sigma\in\mathfrak P}c^{\mathsf T}v^\sigma.
\label{eq:priority-oracle-correctness}
\end{equation}
For \(m\ge1\), sorting requires \(O(m\log m)\)
comparisons, and successive products give the priority
vector in \(O(m)\) arithmetic operations.

\textit{Step 3: Establish the representation and support bound.}
Step 1 gives
\(\operatorname{conv}\{v^\pi:\pi\in\mathfrak P\}
\subseteq P(g)\).
Conversely, if a point of \(P(g)\) lay outside this
convex hull, a separating hyperplane would contradict
\eqref{eq:priority-oracle-correctness}.
The convex hull is compact because \(\mathfrak P\)
is finite. Thus,
\(P(g) = \operatorname{conv}\{v^\pi:\pi\in\mathfrak P\}\).

The Hall inequalities state exactly that \(q\in P(g)\).
Carathéodory's theorem therefore gives a convex
decomposition of \(q\) using at most \(m+1\) priority
vectors. This conclusion holds for arbitrary real
probabilities, including coordinates equal to zero
or one.

\textit{Step 4: Compute a decomposition for rational data.}
Consider the minimum-mass representation problem
\begin{equation}
\begin{aligned}
\min_{\mu\ge0}\quad&
\sum_{\pi\in\mathfrak P}\mu_\pi\\
\textnormal{s.t.}\quad&
\sum_{\pi\in\mathfrak P}\mu_\pi v^\pi=q.
\end{aligned}
\label{eq:priority-primal-lp}
\end{equation}
Step 3 supplies a feasible solution of total mass one,
so the optimal value \(\tau\) satisfies \(0\le\tau\le1\).
The dual is
\begin{equation}
\begin{aligned}
\max_{c\in\R^m}\quad&
c^{\mathsf T}q\\
\textnormal{s.t.}\quad&
c^{\mathsf T}v^\pi\le1,
\qquad \pi\in\mathfrak P.
\end{aligned}
\label{eq:priority-dual-lp}
\end{equation}

Given \(c\), the greedy oracle from Step 2 returns a
maximizing priority vector. If its value exceeds one,
the corresponding inequality separates \(c\) from the
dual feasible region; otherwise, all dual inequalities
hold. For rational \(\eta\), every priority vector has
rational entries of polynomial description size,
because each entry is a product of at most \(m\)
input factors. Coordinates with \(\eta_i=0\) can be
omitted, since their Hall singleton inequalities
force \(q_i=0\).

The constructive ellipsoid and primal-recovery procedure
for rational linear programs
\citep{grotschel1981ellipsoid} therefore solves
\eqref{eq:priority-primal-lp} using polynomially many
greedy-oracle calls. As in the proof of
Proposition~\ref{prop:sblp-polytime-general}, the
procedure recovers an optimal primal solution using
polynomially many oracle-generated columns.
Solving the restricted primal and extracting an
optimal basic solution gives at most \(m\) positive
weights, because there are \(m\) equality constraints.
The empty rule has zero vector and positive objective
coefficient, so it receives zero weight in an optimum
of \eqref{eq:priority-primal-lp}.

Let \(\mu\) be this basic solution and write
\(\tau=\sum_\pi\mu_\pi\). If \(\tau<1\), append the
empty rule with weight \(1-\tau\). Its vector is zero,
so this addition preserves the represented point \(q\).
The resulting weights sum to one and use at most
\(m+1\) rules. Moreover,
\(\sum_h\lambda_hv_0^{\pi^h} = \sum_h\lambda_h \left(1-\sum_{i=1}^m
v_i^{\pi^h}\right) = 1-\sum_{i=1}^m q_i = q_0\).
Thus, the same mixture also gives the correct
no-purchase probability.

Algorithm~\ref{alg:priority-decomposition} summarizes
the construction. The irrational scalar gate used by
the Markov policy is handled by the rational reduction
in Subsection~\ref{app:exact-gate-computation}.\hfill\(\blacksquare\)
\endproof

\begin{algorithm}[t]
\caption{Priority-Rule Decomposition for Rational Inputs}
\label{alg:priority-decomposition}
\begin{algorithmic}[1]
\footnotesize
\REQUIRE Rational \(q\ge0\) and \(\eta\in[0,1]^m\)
satisfying \eqref{eq:hall-condition}.
\ENSURE A decomposition \eqref{eq:priority-decomposition}
with \(L\le m+1\).

\IF{\(q=0\)}
    \STATE \textbf{return} the empty priority rule
    with weight one.
\ENDIF

\STATE Define \(\mathsf{GreedyVertex}(c)\) by listing
the positive-weight products in nonincreasing order
and computing their priority vector by successive
products.

\STATE Use \(\mathsf{GreedyVertex}\) as a separation
oracle for \eqref{eq:priority-dual-lp}, as established
in \eqref{eq:priority-oracle-correctness}.

\STATE Apply the constructive ellipsoid and
primal-recovery procedure to obtain polynomially
many columns supporting an optimum of
\eqref{eq:priority-primal-lp}.

\STATE Solve the restricted primal and extract an
optimal basic solution \(\mu\).

\STATE Assign each positive-weight rule its weight
\(\lambda_\pi\leftarrow\mu_\pi\), and set
\(\tau\leftarrow\sum_\pi\mu_\pi\).

\IF{\(\tau<1\)}
    \STATE Append the empty rule with weight \(1-\tau\).
\ENDIF

\STATE \textbf{return} the rules with positive weights.
\end{algorithmic}
\end{algorithm}

\subsection{Exact Computation of the Fixed-Gate Coupling}
\label{app:exact-gate-computation}

The fixed gate \(\gamma=1-1/e\) introduces an irrational
scalar into the purchase probabilities. We compute the
priority decomposition through a rational LP and then
implement the resulting sampling probabilities exactly.

Fix a nonterminal period, the current virtual states,
realized proposals, and preselected assortment.
Index the visible products by \([m]\), and define
\(b_i:=p_t\one\{i\in\widehat S_t\}\phi_{ti}(\widehat S_t)\)
and \(\eta_i:=1-(1-p_ta_{ti})^{m_{ti}}\).
Under rational model primitives and \(C_{\rm tot}\le MT\),
both vectors are rational and computable in polynomial time. The gated purchase
vector is \(q=\gamma b\).

If \(b=0\), set \(\tau=0\) and use the empty priority
rule with weight one.
Otherwise, solve \eqref{eq:priority-primal-lp} with
right-hand side \(b\), using the greedy separation
oracle and primal-recovery procedure from
Proposition~\ref{prop:priority-decomposition}.
This rational LP is feasible: regularity gives
\(b_i\le p_ta_{ti}\le\eta_i\), and each positive
coordinate \(b_i e_i\) can be represented by the
singleton priority vector \(\eta_i e_i\) with weight
\(b_i/\eta_i\).
Let \(\mu\) be an optimal basic solution supported
on at most \(m\) nonempty rules, and set
\(\tau:=\sum_h\mu_h\).
The minimum representation mass of \(\gamma b\)
is \(\gamma\tau\).
Proposition~\ref{prop:universal-gamma} gives
\(\gamma b\in P(g)\), so \(\gamma\tau\le1\).
Consequently, the weights
\(\lambda_h:=\gamma\mu_h\) and
\(\lambda_{\varnothing}:=1-\gamma\tau\)
form a probability distribution on at most \(m+1\)
priority rules and satisfy
\(\sum_h\lambda_hv^{\pi^h}=q\).
Thus, only rational LPs need to be solved.

For a purchase outcome \(J=i>0\), the posterior
weight of a nonempty rule is
\(\lambda_hv_i^{\pi^h}/q_i
=\mu_hv_i^{\pi^h}/b_i\), which is rational.
For \(J=0\), write \(s:=\sum_i b_i\le1\).
The posterior weights of the nonempty rules and
the empty rule are, respectively,
\(\gamma\mu_hv_0^{\pi^h}/(1-\gamma s)\) and
\((1-\gamma\tau)/(1-\gamma s)\).
Their common denominator satisfies
\(1-\gamma s\ge e^{-1}\).

All required draws admit exact implementations using
unbiased random bits. Certified bounds on \(\gamma\),
and hence on the cumulative sampling probabilities,
with error at most \(2^{-k}\) can be computed in time
polynomial in \(k\) and the input size.
Reveal successive bits of an independent uniform
random number and refine these bounds until its
dyadic interval is certified to lie between two
consecutive cumulative probabilities.
For a distribution with \(d\) outcomes, the probability
of remaining undecided at precision \(k\) is
\(O(d2^{-k})\), so the sampler terminates almost surely
with expected polynomial cost. This also implements
the gate with probability exactly \(1-1/e\).
Under MNL, \(\gamma=1\) and all required sampling
probabilities are rational.
Since \(C_{\rm tot}\le MT\), the full policy therefore
runs in expected time polynomial in \(T,N,M\) and the
model input size, with the stated approximation guarantees.
Arithmetic-operation counts elsewhere treat random
draws as primitive operations.


\subsection{Conditional Sampling of Product-Level Updates}
\label{app:conditional-coupling}

This subsection gives the detailed sampler introduced in
Subsection~\ref{subsec:priority-sampling} and proves
Proposition~\ref{prop:conditional-coupling}.
Fix a priority decomposition satisfying
\eqref{eq:priority-decomposition}.
In the auxiliary joint law, \(U_i\sim\operatorname{Bernoulli}(\eta_i)\)
independently, \(H\) is independent of \(U\) with
\(\Prb(H=h)=\lambda_h\), and
\(J^{\rm f}=\mathsf{sel}_{\pi^H}(U)\).
For an observed purchase outcome \(J=j\) with \(q_j>0\),
the policy samples
\((H,U)\sim\mathcal L(H,U\mid J^{\rm f}=j)\).

\textit{(i) Purchase outcome \(J=i>0\).}
For an observed outcome \(J=i\) with \(q_i>0\), first
draw a priority-rule index \(H\) using the posterior
probabilities
\[
\Prb(H=h\mid J^{\rm f}=i)
=
\frac{\lambda_hv_i^{\pi^h}}{q_i},
\qquad h\in[L].
\]
These probabilities follow from Bayes' rule. Every rule
with positive posterior probability contains \(i\).
Write the selected rule as \(\pi^H=(j_1,\ldots,j_k)\),
and let \(\ell\) be the position of \(i\), so \(j_\ell=i\).

For this rule to select \(i\), all preceding products must
have update indicator zero, and product \(i\) must have
update indicator one. Accordingly, set
\[
U_{j_s}=0\quad(s<\ell),
\qquad
U_i=1.
\]
Finally, independently sample
\(U_j\sim\operatorname{Bernoulli}(\eta_j)\) for every
\(j\in[m]\setminus\{j_1,\ldots,j_\ell\}\).
This includes both products appearing after \(i\) in the
list and products omitted from the rule. Their indicators
retain their independent Bernoulli laws because, for the
selected rule, the event \(J^{\rm f}=i\) constrains only
the prefix through \(i\).

\textit{(ii) No-purchase outcome \(J=0\).}
For an observed outcome \(J=0\) with \(q_0>0\), first
draw \(H\) using the posterior probabilities
\[
\Prb(H=h\mid J^{\rm f}=0)
=
\frac{\lambda_hv_0^{\pi^h}}{q_0},
\qquad h\in[L].
\]
Write the selected rule as \(\pi^H=(j_1,\ldots,j_k)\).
For this rule to return zero, none of its listed products
can have update indicator one. Accordingly, set
\(U_{j_s}=0\) for every \(s\in[k]\).

Finally, independently sample
\(U_j\sim\operatorname{Bernoulli}(\eta_j)\) for every
\(j\in[m]\setminus\{j_1,\ldots,j_k\}\).
Conditional on the selected rule \(H\) and \(J^{\rm f}=0\),
these remaining indicators retain their independent
Bernoulli\((\eta_j)\) distributions. Thus, products outside
the selected list may still receive updates even when no
purchase occurs. The same procedure applies whether \(J=0\)
results from no arrival, a closed gate, or selection of the
outside option.
Algorithm~\ref{alg:conditional-coupling} summarizes these two cases.

\begin{algorithm}[t]
\caption{Conditional Sampling of Product-Level Updates}
\label{alg:conditional-coupling}
\begin{algorithmic}[1]
\footnotesize
\REQUIRE Purchase probabilities \((q_0,\ldots,q_m)\), an
observed outcome \(J\in\{0,\ldots,m\}\) with \(q_J>0\),
update probabilities \(\eta\in[0,1]^m\), and a priority
decomposition \eqref{eq:priority-decomposition}.
\ENSURE A product-level update set \(U\subseteq[m]\).

\IF{\(J>0\)}
    \STATE Set \(i\leftarrow J\), and draw \(H=h\) with
    probability \(\lambda_hv_i^{\pi^h}/q_i\).
    \STATE Write \(\pi^H=(j_1,\ldots,j_k)\), and let
    \(j_\ell=i\).
    \STATE Set \(U_{j_s}\leftarrow0\) for every \(s<\ell\),
    and set \(U_i\leftarrow1\).
    \STATE Independently draw
    \(U_j\sim\operatorname{Bernoulli}(\eta_j)\) for every
    \(j\in[m]\setminus\{j_1,\ldots,j_\ell\}\).
\ELSE
    \STATE Draw \(H=h\) with probability
    \(\lambda_hv_0^{\pi^h}/q_0\).
    \STATE Write \(\pi^H=(j_1,\ldots,j_k)\), and set
    \(U_{j_s}\leftarrow0\) for every \(s\in[k]\).
    \STATE Independently draw
    \(U_j\sim\operatorname{Bernoulli}(\eta_j)\) for every
    \(j\in[m]\setminus\{j_1,\ldots,j_k\}\).
\ENDIF

\STATE \textbf{return} \(U:=\{j\in[m]:U_j=1\}\).
\end{algorithmic}
\end{algorithm}

\proof{Proof of Proposition~\ref{prop:conditional-coupling}.}
For every outcome \(j\) with \(q_j>0\), the posterior weights
are nonnegative and sum to one because
\(\sum_{h=1}^L\lambda_hv_j^{\pi^h}=q_j\).
Conditional on \(H=h\), the event \(J^{\rm f}=i>0\) fixes
exactly the prefix indicators through \(i\), while the event
\(J^{\rm f}=0\) fixes exactly the indicators listed in \(\pi^h\).
All other indicators retain their independent Bernoulli laws.
The two cases above therefore sample precisely from
\(\mathcal L(H,U\mid J^{\rm f}=j)\).

Write \(\mathbf U:=(U_1,\ldots,U_m)\).
Since \(J\) and \(J^{\rm f}\) have the same distribution,
for every \(j\) with \(q_j>0\), every \(h\in[L]\), and every
\(u\in\{0,1\}^m\), the implemented joint law satisfies
\begin{align*}
\Prb(J=j,H=h,\mathbf U=u)
&=q_j\Prb(H=h,\mathbf U=u\mid J^{\rm f}=j)\\
&=\Prb(J^{\rm f}=j,H=h,\mathbf U=u).
\end{align*}
Outcomes with \(q_j=0\) have zero probability under both laws.
Thus, the sampler reproduces the auxiliary joint distribution.
Every purchase belongs to \(U\), and averaging over \(J\)
recovers the prescribed independent Bernoulli update law,
proving both properties.

For the computational bounds, each vector \(v^{\pi^h}\),
including \(v_0^{\pi^h}\), can be computed in \(O(m)\)
arithmetic operations by scanning the rule and maintaining
the prefix product. Computing all posterior weights therefore
takes \(O(mL)\) operations.
For one observed outcome, drawing \(H\) by scanning its posterior
weights takes \(O(L)\) operations. Fixing the constrained
indicators and sampling the remaining ones takes \(O(m)\)
operations, giving \(O(m+L)\) operations per sample.
The bounds \(O(m^2)\) and \(O(m)\) follow from \(L\le m+1\).
\hfill\(\blacksquare\)
\endproof

\section{Proofs of Results in Section~\ref{sec:performance}}
\label{sec:proof-performance}

The universal-gate bound is proved in
Subsection~\ref{subsec:gamma-coefficient}. We first establish its
sharpness and the lossless MNL coupling, then verify the MNL
assortment oracle, copy-update law, and state invariants.
We prove the two approximation guarantees together and give
their associated MNL tightness result, followed by the stronger
zero-outside-option guarantee.

\subsection{Proof of Proposition~\ref{prop:gamma-sharp}}

\proof{Proof.}
\textit{Step 1: Construct the choice and update probabilities.}
For \(m\ge1\), let the customer start uniformly
on \(m\) product states, each of which transitions
directly to the outside option.
Then \(\phi_i(S)=1/m\) for \(i\in S\) and
\(a_i=1/m\).
The product-state submatrix is zero and
therefore transient.

Take \(p=1\), select \(S=[m]\), and suppose
one copy of each product proposes.
The gated purchase probabilities sum to
\(\gamma\).
Each proposed copy is updated with probability
\(1/m\), independently under the target law.
Thus, the probability of at least one update is
\(g_m:=1-\left(1-\frac1m\right)^m\).

\textit{Step 2: Apply the Hall inequality.}
The Hall condition for \(A=[m]\) requires
\(\gamma\le g_m\).
Since \(g_m\to1-1/e\), any fixed
\(\gamma\in(1-1/e,1]\) violates this condition
for sufficiently large \(m\).
The required coupling then does not exist.

\textit{Step 3: Realize the proposal profile in a nonterminal period.}
First consider a one-period instance with
one service slot, one initially free copy
per product, and unit revenues.
The full-assortment SBLP solution has
\(z_i=a_i=1/m\) and objective one.
It is optimal because the total purchase
probability cannot exceed one.
Its proposal probabilities are
\(\rho_i=z_i/a_i=1\), so every copy proposes.

To make this period nonterminal, append a
second request for the same service slot,
with the same arrival and choice parameters
and zero revenues.
The interval balances give
\(x_{2,i,1}([1,1])=1-1/m\).
In the second period, select the empty
assortment and set all sales flows to zero.
This extends the first-period solution to
a feasible SBLP solution of value one.
It remains optimal because the second period
generates no revenue.

The first period therefore realizes the stated
proposal profile with probability one.
With unit revenues, the assortment oracle
selects all products, so the Hall obstruction
arises within the online construction.
Together with Proposition~\ref{prop:universal-gamma},
this proves sharpness among uniform scalar gates.\hfill\(\blacksquare\)
\endproof

\subsection{Proof of Proposition~\ref{prop:mnl-gamma-one}}

\proof{Proof.}
Fix an arrival probability \(p\in(0,1]\), an
assortment \(S\), and a subset \(A\subseteq S\).
The case \(A=\varnothing\) is immediate.
For \(A\ne\varnothing\), write
\(U_A:=\sum_{i\in A}u_i\) and
\(U_S:=\sum_{i\in S}u_i\).
The MNL choice probabilities satisfy
\begin{equation}
p\sum_{i\in A}\phi_i(S)
=
\frac{pU_A}{1+U_S}
\le
\frac{pU_A}{1+U_A}.
\label{eq:mnl-choice-mass}
\end{equation}

Define
\(H_A(s):=1-\prod_{i\in A}(1-sa_i)\)
for \(s\in[0,1]\), where
\(a_i=u_i/(1+u_i)\).
This function is concave because
\[
H_A''(s)
=
-2\sum_{\substack{i,j\in A\\i<j}}
a_i a_j
\prod_{k\in A\setminus\{i,j\}}(1-sa_k)
\le0.
\]
For a singleton \(A\), the sum is empty and
\(H_A\) is linear.
Since \(H_A(0)=0\), concavity gives
\(H_A(p)\ge pH_A(1)\).
Furthermore,
\[
\begin{aligned}
H_A(1)
=
1-\frac{1}{\prod_{i\in A}(1+u_i)}\ge
1-\frac{1}{1+U_A}
=
\frac{U_A}{1+U_A},
\end{aligned}
\]
where the inequality follows from
\(\prod_{i\in A}(1+u_i)\ge1+\sum_{i\in A}u_i\).

Combining these bounds with
\eqref{eq:mnl-choice-mass} yields
\(1-\prod_{i\in A}(1-pa_i) = H_A(p) \ge \frac{pU_A}{1+U_A} \ge p\sum_{i\in
A}\phi_i(S)\).
Thus, every Hall inequality holds without gating
when each visible product has one proposed copy.
Additional proposed copies can only increase
update coverage.

Every positive-denominator ratio in
\eqref{eq:gamma-choice} is therefore at least
one. Since the definition of \(\Gamma(\phi)\)
caps its value at one, \(\Gamma(\phi)=1\).\hfill\(\blacksquare\)
\endproof

\subsection{Efficient MNL Assortment Optimization}

Consider a candidate set \(A\), nonnegative
product weights \(w_i\), and positive outside
attractiveness normalized to one.
Write
\(R_w(S):= \frac{\sum_{i\in S}w_i u_i} {1+\sum_{i\in S}u_i}\) and
\(R^*:=\max_{S\subseteq A}R_w(S)\).
For \(j\in A\setminus S\), adding product \(j\)
changes the objective by
\[
R_w(S\cup\{j\})-R_w(S)
=
\frac{u_j(w_j-R_w(S))}
     {1+\sum_{i\in S}u_i+u_j}.
\]
For \(j\in S\), removing it changes the objective by
\[
R_w(S\setminus\{j\})-R_w(S)
=
\frac{u_j(R_w(S)-w_j)}
     {1+\sum_{i\in S}u_i-u_j}.
\]

At an optimum, every product with \(w_i>R^*\)
must be included, and every product with
\(w_i<R^*\) must be excluded.
Products with \(w_i=R^*\) can be added without
changing the objective.
Thus, an optimal assortment can be chosen as
a prefix of the products ordered by
nonincreasing weight.

After sorting, cumulative sums of \(u_i\)
and \(w_i u_i\) allow all prefixes to be
evaluated in linear time.
Including the empty assortment, the oracle
requires \(O(|A|\log|A|)\) arithmetic and
comparison operations; an empty candidate
set is handled directly.
This implements both the actual-revenue oracle
and the nonnegative weighted queries used
in solving the SBLP.

\subsection{Copy-Level Sampling and Proof of Lemma~\ref{lem:copy-update-law-general}}
\label{app:copy-update-law}

\proof{Proof.}
Fix a nonterminal period \(t<T\), the current virtual
states, realized proposals, and preselected assortment.
All probabilities below retain this conditioning;
the gate and customer outcomes remain random.
For \(i\in P_t\), write
\(m_i:=m_{ti}\),
\(\alpha_i:=p_ta_{ti}\), and
\(\eta_i:=1-(1-\alpha_i)^{m_i}\).

By Proposition~\ref{prop:conditional-coupling},
averaging over \(J_t\) gives mutually independent
\(U_i\sim\operatorname{Bernoulli}(\eta_i)\).

\textit{Step 1: Recover the update law within each product.}
For \(b\in\{0,1\}^{m_i}\), let
\(\nu_i(b):= \prod_{\ell:b_\ell=1}\alpha_i
\prod_{\ell:b_\ell=0}(1-\alpha_i)\)
denote its probability under the independent
Bernoulli\((\alpha_i)\) law.
The zero vector has probability
\(\nu_i(0)=(1-\alpha_i)^{m_i}=1-\eta_i\).

The lifting returns the zero vector when \(U_i=0\).
When \(U_i=1\), it returns a nonzero vector \(b\)
with conditional probability \(\nu_i(b)/\eta_i\).
Thus, after averaging over \(U_i\), every vector
has probability \(\nu_i(b)\).
If \(\eta_i=0\), then \(\alpha_i=0\) and \(U_i=0\)
almost surely, so the vector is deterministically
zero and no division by \(\eta_i\) is needed.

\textit{Step 2: Establish independence across products and copies.}
Write \(X_i=(X_{i1},\ldots,X_{im_i})\), and let
\(L_i(b_i\mid u_i)\) denote its conditional lifting
law, with arbitrary values on conditioning events
of probability zero. The lifting uses independent
random draws across products and depends on the
purchase outcome only through \(U\).
Therefore,
\[
\begin{aligned}
\Prb(X_i=b_i\text{ for all }i\in P_t)
&=
\sum_{u\in\{0,1\}^{P_t}}
\prod_{i\in P_t}
\Prb(U_i=u_i)L_i(b_i\mid u_i)\\
&=
\prod_{i\in P_t}
\left[
\sum_{u_i=0}^1
\Prb(U_i=u_i)L_i(b_i\mid u_i)
\right]\\
&=
\prod_{i\in P_t}\nu_i(b_i).
\end{aligned}
\]
The first equality uses independence of the
product-level indicators and independent lifting
conditional on them. The final equality follows
from Step 1. Since each \(\nu_i\) is itself a
Bernoulli product law, the resulting joint
distribution factors over all proposed copies.

\textit{Step 3: Sample a nonempty copy vector directly.}
For \(0<\alpha_i<1\) and \(U_i=1\), draw the position
of the first success according to
\[
\Prb(K_i=k\mid U_i=1)
=\frac{\alpha_i(1-\alpha_i)^{k-1}}{\eta_i},
\qquad k\in[m_i].
\]
These probabilities sum to one because
\(\sum_{k=1}^{m_i}\alpha_i(1-\alpha_i)^{k-1}
=1-(1-\alpha_i)^{m_i}=\eta_i\).

Set all coordinates before \(K_i\) to zero, set
coordinate \(K_i\) to one, and sample the remaining
coordinates independently as
Bernoulli\((\alpha_i)\).
For any nonzero vector \(b\), multiplying the
first-success probability by the probabilities
of its remaining coordinates gives
\(\nu_i(b)/\eta_i\), as required.

If \(\alpha_i=1\), the vector is deterministically
all ones. If \(\alpha_i=0\), the branch \(U_i=1\)
never occurs. The construction requires
\(O(m_i)\) arithmetic operations and random draws,
without a rejection step whose cost depends on
\(\eta_i\).\hfill\(\blacksquare\)
\endproof

\subsection{Proof of Theorem~\ref{thm:implementation-general}}

\proof{Proof.}
Fix the SBLP solution used by the policy.
We prove the virtual-state properties by induction
using the joint transition law, and establish
physical feasibility along each sample path.
The terminal period is treated separately.

\textit{Step 1: Initialize the induction.}
At \(t=1\), every actual and virtual free-slot set
equals the deterministic set \(\cN\).
Thus, virtual availability is contained in actual
availability, the interval marginals satisfy
\eqref{eq:sblp-boundary}, and the virtual states
are independent across copies.

\textit{Step 2: Derive the joint transition law in a
nonterminal period.}
Suppose the assertions hold at the beginning of
period \(t<T\). Index copies by \(k=(i,c)\), and
write \(\overline{\mathcal F}_t\) for the entire
virtual-state profile.
Fix a realization \(F=(F_k)_k\) of this profile.

Let \(B_k\) denote copy \(k\)'s proposal indicator.
Its conditional proposal probability
\(\beta_k(F_k)\) equals \(\pi_{tic}(K)\) if \(F_k\)
has a maximal free interval \(K\) containing \(I_t\),
and equals zero otherwise.
Conditional on \(F\), the proposal indicators are
independent, with
\(\Prb(B_k=1\mid F)=\beta_k(F_k)\).

Let \(D_k\) indicate whether copy \(k\) receives
a virtual cut, and write
\(\alpha_k:=p_ta_{ti}\) for \(k=(i,c)\).
A nonproposed copy is never cut.
For a proposed copy,
Lemma~\ref{lem:copy-update-law-general} gives the
target Bernoulli update law after averaging over
the gate and customer outcome.
Thus,
\(\Prb(D=d\mid F,B=b) = \prod_k L_k(d_k\mid b_k)\),
where
\(L_k(1\mid0)=0\),
\(L_k(0\mid0)=1\),
\(L_k(1\mid1)=\alpha_k\), and
\(L_k(0\mid1)=1-\alpha_k\).
The lemma holds for every preselected assortment,
and this product law does not depend on which
assortment is selected. The identity therefore
also holds if the assortment oracle randomizes
among maximizers.

Each \(L_k\) depends on the proposal profile only
through \(b_k\). Averaging over the independent
proposals consequently gives
\[
\begin{aligned}
\Prb(D=d\mid F)
&=
\sum_{b\in\{0,1\}^{C_{\rm tot}}}
\prod_k
\Prb(B_k=b_k\mid F_k)L_k(d_k\mid b_k)\\
&=
\prod_k
\sum_{b_k\in\{0,1\}}
\Prb(B_k=b_k\mid F_k)L_k(d_k\mid b_k).
\end{aligned}
\]
Hence, the cut indicators are independent
conditional on \(F\), with success probabilities
\(\delta_k(F_k):=\beta_k(F_k)\alpha_k\).
For a virtually feasible copy whose containing
maximal interval is \(K\),
\(\delta_k(F_k) = \pi_{tic}(K)p_ta_{ti} = \frac{y_{tic}(K)}{x_{tic}(K)} =
\kappa_{tic}(K)\),
using the zero-ratio convention when
\(x_{tic}(K)=0\).
Virtually infeasible copies have zero cut probability.

\textit{Step 3: Preserve independence and interval marginals.}
The next virtual state of copy \(k\) is
\(F_k\setminus I_t\) if \(D_k=1\), and \(F_k\)
otherwise. Its transition therefore depends
only on its own current state, and the joint
transition kernel factors across copies.

Let \(K_k(F'_k\mid F_k)\) denote this one-copy
transition kernel and \(\mu_{t,k}\) the current
one-copy state distribution.
Independence at the beginning of period \(t\)
then gives
\[
\begin{aligned}
\Prb(\overline{\mathcal F}_{t+1}=F')
&=
\sum_F
\prod_k
\mu_{t,k}(F_k)K_k(F'_k\mid F_k)\\
&=
\prod_k
\sum_{F_k}
\mu_{t,k}(F_k)K_k(F'_k\mid F_k).
\end{aligned}
\]
This proves independence at the next period boundary.
The local transition agrees with
Lemma~\ref{lem:copy-flow-realizability}, so it also
preserves the prescribed interval marginals
\(x_{t+1,i,c}(K)\).

Iterating this argument establishes properties
(iii) and (iv) through the beginning of period \(T\).
These distributional statements average over past
outcomes; they do not condition on a realized
purchase history.

\textit{Step 4: Preserve physical feasibility and state inclusion.}
Suppose virtual availability is contained in
actual availability at the beginning of period \(t<T\).
Every proposed copy then satisfies
\(I_t\subseteq\overline{\mathcal F}_{tic}
\subseteq\mathcal F_{tic}\), so it can physically
fulfill the request.

If \(J_t=i\), the coupling gives \(U_{ti}=1\),
and the lifting ensures \(W_{ti}\ne\varnothing\).
The policy assigns the purchase to one copy
\(c^*\in W_{ti}\). Removing \(I_t\) from both
sets of this copy preserves inclusion:
\(\overline{\mathcal F}_{t,i,c^*}\setminus I_t \subseteq \mathcal
F_{t,i,c^*}\setminus I_t\).
Every other updated copy loses \(I_t\) only from
its virtual set, and all remaining copies are
unchanged. If \(J_t=0\), the actual sets remain
unchanged and the virtual sets can only shrink.

Thus, each nonterminal period preserves virtual
containment, and every purchase reserves the
requested interval on exactly one physically
feasible copy.

\textit{Step 5: Complete the terminal period.}
At the beginning of period \(T\), the preceding
induction has established all state properties.
Any purchased product belongs to the visible set
and therefore has a proposed, physically feasible
copy. The terminal rule assigns the purchase to
one such copy and removes the requested interval
from both its actual and virtual sets.

This establishes physical feasibility throughout
the horizon and virtual containment through
period boundary \(T+1\). No prescribed interval
marginals or independence property at \(T+1\)
are required, completing the proof.\hfill\(\blacksquare\)
\endproof

\subsection{Proof of Theorem~\ref{thm:main-guarantees}}
\label{subsec:proof-performance-markov}

\proof{Proof.}
Let \(\pi\) denote Algorithm~\ref{alg:general-policy} with
the exact assortment oracle for the relevant choice model
and gate probability \(\gamma\) in every nonterminal period.
Theorem~\ref{thm:implementation-general} establishes feasibility
and ensures that the static extraction bound applies in every
period.

Set \(W_t:=\sum_{i\in\cM}r_{ti}z_{ti}\), and let
\(P_t\) denote the visible product set in period \(t\).
By Corollary~\ref{cor:visible-static-general},
\(\E[f_t(P_t)]\ge W_t/2\).
The coupling preserves the purchase probabilities induced
by the selected gate, and the terminal period requires no
gate. Multiplying by the arrival probabilities and summing
over periods therefore gives
\[
\begin{aligned}
V^\pi(I)
&=\gamma\sum_{t<T}p_t\E[f_t(P_t)]
  +p_T\E[f_T(P_T)]\\
&\ge\frac12\left(\gamma\sum_{t<T}p_tW_t+p_TW_T\right)\\
&\ge\frac{\gamma}{2}\sum_{t=1}^T p_tW_t
 =\frac{\gamma}{2}\SBLP^C(I).
\end{aligned}
\]
The final inequality uses \(0\le\gamma\le1\) and
nonnegative revenues.

Under general Markov chain choice,
Proposition~\ref{prop:universal-gamma} makes the coupling
feasible with \(\gamma=1-1/e\), giving
\(V^{\pi_{\rm MC}}(I)\ge(1-1/e)\SBLP^C(I)/2\).
Under MNL, Proposition~\ref{prop:mnl-gamma-one} permits
\(\gamma=1\), giving
\(V^{\pi_{\rm MNL}}(I)\ge\SBLP^C(I)/2\).
Since the SBLP upper-bounds \(V^{\rm on}(I)\), these
inequalities establish \eqref{eq:markov-main} and
\eqref{eq:mnl-main}, respectively.

The relaxation, assortment optimization, and priority
decompositions are computable in polynomial time.
Our convention \(C_i\le T\) gives \(C_{\rm tot}\le MT\), so
maintaining individual copies also requires polynomial space
and time. The required random draws admit exact implementations
with expected polynomial running time, as detailed in
Appendix~\ref{app:exact-gate-computation}.
\hfill\(\blacksquare\)
\endproof

\subsection{Proof of Proposition~\ref{prop:mnl-policy-tightness}}

\proof{Proof.}
For each \(m\ge2\), take the MNL choice model constructed
in Appendix~\ref{app:proof-markov-extraction-sharpness}
and form an instance \(\mathcal I_m\) with one booking
period, one service slot, arrival probability one, and
one physical copy per product. Let \(W_m\) denote the
full-assortment revenue, and let \(P_m\) and \(f_m\)
denote the independent candidate set and optimal
assortment-revenue function from that construction.
Since revenue is strictly increasing in assortment size,
the full assortment is the unique optimal offer, and
\(V^{\mathrm{on}}(\mathcal I_m)=W_m\).

Suppressing the period and copy indices, the full-assortment
mixture induces the feasible SBLP solution
\[
x_i([1,1])=1,
\qquad
y_i([1,1])=z_i=\phi_i(\cM).
\]
The state-capacity constraint holds because \(z_i\le a_i\).
Every feasible SBLP objective is a convex combination of
assortment revenues and therefore cannot exceed \(W_m\).
Strict optimality of the full assortment forces every
optimal mixture to assign it weight one. With one copy
and one service interval per product, the sales flows
are consequently unique.

The policy therefore proposes each product independently
with probability
\[
\frac{y_i([1,1])}{a_i x_i([1,1])}
=\frac{z_i}{a_i},
\]
so its visible set has exactly the law of \(P_m\),
regardless of how an optimal SBLP solution is selected.
The single period is ungated, and the assortment oracle
earns \(f_m(P_m)\) in conditional expectation. Hence,
\[
\frac{V^{\pi_{\mathrm{MNL}}}(\mathcal I_m)}
     {V^{\mathrm{on}}(\mathcal I_m)}
=
\frac{\E[f_m(P_m)]}{W_m}
\longrightarrow \frac12,
\]
where the limit was established in
Appendix~\ref{app:proof-markov-extraction-sharpness}.
The first member \(\mathcal I_1\) may be any admissible
one-period, unit-capacity MNL instance with positive
optimal revenue.
\hfill\(\blacksquare\)
\endproof

\subsection{Verification of Remark~\ref{rem:mnl-zero-outside}}\label{app:mnl-zero-outside}

\proof{Proof.}
\textit{Step 1: Define the boundary choice model.}
Fix a period and suppress its index. With zero outside-option
attractiveness and \(u_i>0\) for every product, define
\[
\phi_i(S)=
\begin{cases}
\displaystyle\frac{u_i}{\sum_{j\in S}u_j}, & i\in S,\\[2mm]
0, & i\notin S,
\end{cases}
\qquad
\phi_0(S)=\one\{S=\varnothing\}.
\]
The denominator is used only when \(i\in S\), so it is
positive. Every nonempty assortment has total purchase
probability one, and \(a_i=1\). Regularity follows directly
from the choice formula.

The revenue of a nonempty assortment is an
attractiveness-weighted average of its product revenues.
Consequently, \(f(A)=\max_{i\in A}r_i\) for \(A\ne\varnothing\),
with \(f(\varnothing)=0\); a highest-revenue singleton attains
this value.

\textit{Step 2: Establish the stronger static extraction bound.}
Consider the independent copy proposals generated from the SBLP
interval marginals, and let \(P_{\rm vis}\) be their visible
product set. Lemma~\ref{lem:aggregate-identities} and \(a_i=1\)
give \(\sum_c\rho_{ic}=z_i\), while the choice-mixture
constraints give \(\sum_i z_i\le1\).

For a revenue threshold \(\tau\ge0\), define
\(A_\tau:=\{i:r_i>\tau\}\) and
\(Z_\tau:=\sum_{i\in A_\tau}z_i\in[0,1]\).
By independence across copies,
\[
\begin{aligned}
\Prb\bigl(f(P_{\rm vis})>\tau\bigr)
&=1-\prod_{i\in A_\tau}\prod_{c=1}^{C_i}(1-\rho_{ic})\\
&\ge1-\exp\!\left(-\sum_{i\in A_\tau}\sum_{c=1}^{C_i}\rho_{ic}\right)\\
&=1-e^{-Z_\tau}
\ge\left(1-\frac1e\right)Z_\tau.
\end{aligned}
\]
The first inequality uses \(1-u\le e^{-u}\); the last follows
from concavity of \(1-e^{-z}\) on \([0,1]\). Integrating
over revenue thresholds yields
\[
\begin{aligned}
\E[f(P_{\rm vis})]
&=\int_0^\infty\Prb\bigl(f(P_{\rm vis})>\tau\bigr)\,d\tau\\
&\ge\left(1-\frac1e\right)\int_0^\infty Z_\tau\,d\tau
=\left(1-\frac1e\right)\sum_i r_i z_i.
\end{aligned}
\]
This bound holds for arbitrary copy multiplicities and does
not require the copy-level flows to be symmetric.

\textit{Step 3: Verify the coupling and online guarantee.}
For every arrival probability \(p\in(0,1]\) and nonempty
\(A\subseteq S\),
\(1-\prod_{i\in A}(1-pa_i)=1-(1-p)^{|A|}\ge p
\ge p\sum_{i\in A}\phi_i(S)\).
The empty-subset inequality is automatic. Hence, all Hall
conditions hold with \(\gamma=1\); additional proposed copies
only increase the update coverage.

The SBLP upper-bound and interval-flow arguments remain valid
for these regular choice probabilities and singleton bounds.
The posterior sampler and copy-level lifting reproduce the
prescribed independent update law, so the induction in
Theorem~\ref{thm:implementation-general} also applies. Thus,
the proposal distribution used in Step 2 is maintained in
every period, and every purchase is physically feasible.

Restoring the time index, set \(W_t:=\sum_i r_{ti}z_{ti}\).
If every period has zero outside-option attractiveness,
all periods are ungated and the stronger static bound gives
\[
\begin{aligned}
V^{\pi_{\rm MNL}}(I)
&=\sum_{t=1}^T p_t\E[f_t(P_t)]\\
&\ge\left(1-\frac1e\right)\sum_{t=1}^T p_tW_t
=\left(1-\frac1e\right)\SBLP^C(I)
\ge\left(1-\frac1e\right)V^{\rm on}(I).
\end{aligned}
\]

\textit{Step 4: Verify tractability and the scope of the improvement.}
For arbitrary signed product weights, a singleton attaining
the largest positive eligible weight is optimal. If no eligible
weight is positive, the empty assortment is optimal. This
provides a polynomial-time weighted assortment oracle.
The aggregate relaxation, symmetric lift, and conditional
sampling procedures therefore remain tractable.

If positive and zero outside-option attractiveness occur in
different periods, the same policy remains feasible. Its
period-level extraction factor is \(1-1/e\) in each zero-outside
period and \(1/2\) in each positive-outside period. In particular,
the general \(1/2\) guarantee still holds in this mixed setting;
the uniform \(1-1/e\) guarantee above concerns instances with
zero outside-option attractiveness in every period.
\hfill\(\blacksquare\)
\endproof

\textit{Tightness under the symmetric lift.}
Consider \(T=C\) booking periods, one service slot, and
one product with \(C\) copies. In every period, arrival
is certain, the request is for the single service slot,
and the outside option has zero attractiveness.
Revenue is one in the first period and zero in all
remaining periods. Every optimal SBLP solution has
first-period sales probability one, so the symmetric
lift assigns first-period proposal probability \(1/C\)
to each copy. The policy's expected revenue is therefore
\(1-(1-1/C)^C\), whereas \(V^{\rm on}(I)=1\).
The ratio approaches \(1-1/e\) as \(C\to\infty\),
showing that the factor is tight for the policy using
the symmetric lift.

\end{APPENDICES}

\end{document}